\documentclass[12pt]{article}
\usepackage[english]{babel}
\usepackage{pgf, amsmath, amsfonts, amsthm,amsbsy, mathrsfs, latexsym, multicol,amssymb}
\usepackage{srcltx, hyperref, stackrel, caption}
  \usepackage[rightcaption]{sidecap}
  
  \usepackage[margin=3cm]{geometry}

\newcommand{\epsi}{\varepsilon}
\newcommand{\E}{{\mathrm{e}}}
\newcommand{\I}{\mathrm{i}}
 \newcommand{\R}{ \mathbb{R} }
\newcommand{\C}{ \mathbb{C} }
\newcommand{\N}{ \mathbb{N} }
\newcommand{\Z}{ \mathbb{Z} }
\newcommand{\D}{\mathrm{d}}
\newcommand{\Or}{{\mathcal{O}}}

\newcommand{\tr}{{\mathrm{tr}}}
\newcommand{\Hi}{{\mathfrak{H}}}
\newcommand{\bm}{\begin{pmatrix}}
\newcommand{\Em}{\end{pmatrix}}
  
\newcommand{\add}{\mathrm{ad}}
\newcommand{\ad}{\mathscr{L}}

\newcommand{\dege}{\kappa}
\newcommand{\loc}{{\rm Loc}}

\newcommand{\rom}{\renewcommand{\labelenumi}{{\rm(\roman{enumi})}}}

\newcommand{\alp}{\renewcommand{\labelenumi}{{\rm(\alph{enumi})}}}
\newcommand{\balp}{\begin{enumerate}\alp}
\newcommand{\ealp}{\end{enumerate}}

\newcommand{\brom}{\begin{enumerate}\rom}
\newcommand{\erom}{\end{enumerate}}
 
 \newcommand{\lminus}{%
  \mathrel{\vbox{\offinterlineskip\ialign{%
    \hfil##\hfil\cr
    $\scriptscriptstyle\Lambda$\cr
    \noalign{\kern-0.2ex}
    $-$\cr
}}}}

\newcommand{\Mminus}{%
  \mathrel{\vbox{\offinterlineskip\ialign{%
    \hfil##\hfil\cr
    $\scriptscriptstyle{M}$\cr
    \noalign{\kern-0.2ex}
    $-$\cr
}}}}
 
\newcommand{\SVP}{\mathcal{V}}
\newcommand{\dimfib}{s}
\newcommand{\timedom}{I}

   \newtheorem{theorem}{Theorem}[section]
  \newtheorem{definition}{Definition}[section]
\newtheorem*{theorem*}{Theorem}
\newtheorem{lemma}{Lemma}[section]
\newtheorem*{lemma*}{Lemma}

\newtheorem{proposition}{Proposition}[section]
\newtheorem{corollary}{Corollary}[section]

\newtheorem{remark}{Remark}[section]

\title{Non-equilibrium almost-stationary states  and\\ linear response for   gapped quantum systems}
\author{ Stefan Teufel\thanks{
Fachbereich Mathematik, Eberhard-Karls-Universit\"at\newline
\textcolor{white}{a} \hspace{.7em} Auf der Morgenstelle 10, 72076 T\"ubingen, Germany\newline
\textcolor{white}{a} \hspace{.7em} E-mail:   stefan.teufel@uni-tuebingen.de}}
 
\begin{document}
\maketitle

\begin{abstract} 
We prove the validity of linear  response theory at zero temperature for  perturbations of  gapped Hamiltonians describing interacting fermions on a lattice. As an essential innovation, our result requires  
the spectral gap assumption only for the unperturbed Hamiltonian and applies to a large class of  perturbations that close the spectral gap. Moreover, we prove formulas also for higher order response coefficients.

Our justification of linear response theory  is based on a novel
extension of the adiabatic theorem to  situations where a time-dependent  perturbation closes the gap. 
According to the standard version of the adiabatic  theorem, when the perturbation is switched on adiabatically and as long as the gap does not close, the initial ground state evolves into the ground state of the perturbed operator.
The new adiabatic theorem states that for perturbations that are either slowly varying potentials or small quasi-local operators,   once the perturbation closes the gap, the adiabatic evolution follows  \emph{non-equilibrium almost-stationary states} 
(NEASS) that we construct explicitly.

\medskip

\noindent \textbf{Keywords.} Linear response theory, adiabatic theorem, non-equilibrium stationary state, space-adiabatic perturbation theory, Kubo formula.

\medskip

\noindent \textbf{AMS 
Mathematics 
Subject 
Classification (2010).} 81Q15; 
81Q20; 81V70.
\end{abstract}

%
%

 \section{Introduction}
The simplicity and the empirical success of  linear response theory \cite{Ku} make it a formalism widely used in physics to calculate the response of systems in thermal equilibrium to external perturbations. 
However, its validity for extended systems is based on properties of the microscopic dynamics, which are often difficult  to establish. 
It is therefore not surprising that the rigorous justification of linear response theory based on first principles in specific models is a constant theme in mathematical physics, which was prominently advertised for example by Simon \cite{S} already in 1984.

In this work we prove the validity of linear and also higher order response theory for perturbations of gapped interacting  quantum Hamiltonians on the lattice and  at zero temperature. This  framework is relevant, for example, for (topological) insulators in solid state physics such as quantum Hall systems.

More specifically, we consider a family of quasi-local Hamiltonians for  systems of interacting fermions on finite cubes $\Lambda\subset \Z^d$ with a spectral gap above the ground state, whose size is bounded below by a positive constant  uniformly in the volume $|\Lambda|$. 
Then, according to equilibrium statistical mechanics, the equilibrium state of the system at sufficiently low temperature is very close to its ground state.
A question of fundamental physical importance is to understand the ``response'' of such systems to    static perturbations as, for example, a weak external electric field. Here ``response'' refers to the change of expectation values of physical quantities which are induced by adiabatically switching  on the perturbation. Linear response theory proceeds by applying first order time-dependent perturbation theory in an uncontrolled way, cf.\ Section~\ref{sec:LinearResponse}.
Thus, any justification of linear response theory in the present context starts necessarily from the analysis of solutions of the time-dependent Schr\"odinger equation in the adiabatic limit. 
However, the standard adiabatic theorem, which provides asymptotic expansions of these solutions to any order in the adiabatic parameter, falls short for extended interacting systems for two 
reasons. Firstly, it yields norm-estimates that  are not and cannot be uniform in the volume $|\Lambda|$. Secondly, it must be assumed that a spectral gap    remains open uniformly in the volume $|\Lambda|$ even if the perturbation is fully turned on.

Recently, Bachmann et al.\ \cite{BDF} were able to prove an adiabatic theorem for expectation values of local observables in interacting spin systems with error estimates that are uniform in the volume $|\Lambda|$. Their result  was slightly extended and translated   to the setting of lattice fermions in \cite{MT}.
While the result of Bachmann et al.\ is a technical and conceptual breakthrough, it is still an extremely difficult open problem to prove their main assumption, namely the stability of the gap of a generic gapped many-body Hamiltonian under small perturbations. 
More importantly, the linear response formalism is  expected to be applicable also in situations 
where the  perturbation  closes the spectral gap and the system is driven into an (almost) stationary state that need not be an eigenstate.

In this article we formulate and prove a novel adiabatic theorem with a gap assumption only on the   unperturbed Hamiltonian. The class of allowed perturbations contains slowly varying but not necessarily small potentials and small quasi-local operators. 
It is shown that, once the spectral gap closes, the adiabatic evolution no longer  follows the ground state of the system---which is an invariant state for the instantaneous Hamiltonian---but instead  a certain almost-invariant state for the now gapless    Hamiltonian. As these almost-invariant states are neither eigenstates nor functions of the Hamiltonian, we call them {\em non-equilibrium almost-stationary states} (NEASS). Rigorous and uniform asymptotic expansions of the NEASS, into which the system evolves when adiabatically turning on a perturbation, then allow for a straightforward proof of linear and higher order response theory.

Since the formulation of precise statements requires the implementation of a decent amount of not completely standard mathematical concepts, we refrain from stating theorems in the introduction and instead briefly explain the main conceptual ideas behind our proof. 
As realised and  worked out in \cite{BDF}, the key concepts for adiabatic approximations that hold uniformly in the volume are locality and finite speed of propagation in lattice systems. By assumption, all operators appearing (the Hamiltonian, the perturbation, the observables) are quasi-local, i.e.\ they are sums of local terms. 
While the number of summands in these operators increases with increasing volume, 
in any bounded region only finitely many summands make a sizeable contribution.
Moreover, thanks to Lieb-Robinson bounds \cite{LR} no long range correlations are induced by the dynamics and, as a consequence, the 
spectral flow is generated by a quasi-local operator \cite{HW,BMNS}. This allowed Bachmann et al.\ \cite{BDF} to control adiabatic approximation errors for expectations of   local observables uniformly in the volume.

To dispose of the gap assumption   for the perturbed Hamiltonian, in the present article we use  that   small quasi-local perturbations and  slowly varying potentials both leave intact a local gap structure, even though the full perturbed Hamiltonian might not be gapped anymore. For small quasi-local perturbations this is expected, because for any fixed volume stability of the gap follows from standard perturbation theory. Slowly varying potentials, on the other hand, are locally almost constant and thus only shift the local terms in the Hamiltonian by a multiple of the identity. As a consequence, the  NEASS of the perturbed system can be constructed by applying a unitary transformation that is generated by a sum of local terms to the ground state of the unperturbed Hamiltonian.  The resulting state is almost-invariant under the dynamics generated by the perturbed Hamiltonian, because  transitions out of this state either require particles to overcome the local energy gap or to tunnel long spatial distances. 
For a heuristic sketch of the situation see Figure~\ref{fig1}. To implement these ideas mathematically, we heavily  use   and partly extend  a technical machinery  that has experienced important new  developments during recent years.
This includes Lieb-Robinson bounds for interacting fermions \cite{NSY,BD}, the quasi-local inverse of the Liouvillian introduced in the context of the so called spectral-flow or quasi-adiabatic evolution \cite{HW,BMNS}, 
as well as ideas and technical lemmas from \cite{BDF} and \cite{MT}. 
\begin{figure}[h]
\begin{center}
\includegraphics[width=\textwidth]{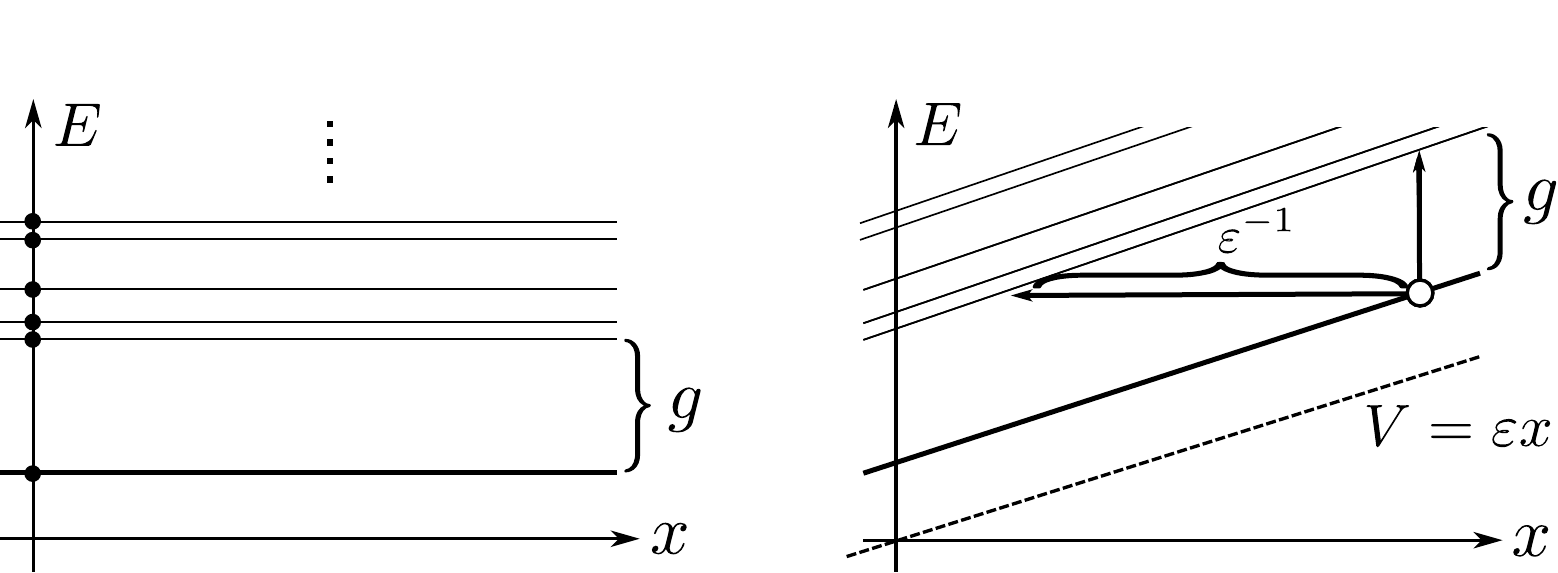}
\end{center}
\caption{On the left panel the eigenvalues of a gapped Hamiltonian $H_0$ are sketched: The ground state is separated by a gap $g$ from the rest of the spectrum uniformly in the volume, while the eigenvalue spacing between the other eigenvalues typically goes to zero when the volume grows. On the right panel the local energy landscape of $H=H_0+V$ with $V=\epsi x$, $\epsi\ll 1$, is sketched and the lowest solid  line represents the NEASS. Local transitions out of the NEASS are strongly suppressed, since either the spectral gap of size $g$ (vertical arrow) or a spatial distance of size $\epsi^{-1}$ (horizontal arrow) needs to be overcome. \label{fig1}}
\end{figure}
The mathematical justification of linear response theory in similar situations for non-interacting systems of fermions was studied e.g.\ in 
\cite{BES,BGKS,DL,ES,KLM}. Note that while the results in \cite{BGKS,KLM,DL} require instead of a spectral gap only a mobility gap, because of the order of limits they are not yet fully satisfactory, cf.\ Section~\ref{sec:LinearResponse}. Clearly linear response in the presence of localisation is of great physical relevance also in the interacting case, however,   the mathematical understanding of many-body localisation is still in its infancy and we are not aware of any work on justifying linear response in this setting.

For non-interacting systems NEASS were constructed e.g.\ in \cite{PST1,PST2,PST3,T}
using  the formalism of space-adiabatic perturbation theory.
In this context  a different  terminology was adopted, and instead of almost-stationary states one speaks about almost-invariant subspaces, a notion  going back to \cite{N}.
The results of the present paper could thus be viewed as a generalisation of space-adiabatic perturbation theory to interacting systems. 

The structure of the paper is as follows.  In Section~\ref{sec:frame} we introduce the mathematical   framework for quasi-local  operators and extend it   to partly localised quasi-local operators. While the concept of slowly varying potentials is not new, its definition for systems of varying size  requires some care. Moreover, we formulate a crucial lemma, Lemma~\ref{lemma:Vcomm}, that states that commutators of arbitrary quasi-local operators with slowly varying potentials are small and quasi-local.
In Section~\ref{sec:NEASS}  we state those results concerning the NEASS that are required for the proof of linear response theory in Section~\ref{sec:LinearResponse}.
    Section~\ref{sec:spacetime} contains the general adiabatic theorem, Theorem~\ref{theorem:Adi}, which we call  a space-time adiabatic theorem, since it 
exploits the slow variation of the Hamiltonian both in space and as a function of time. 
Its proof is divided into two parts.
The proof of the space-time adiabatic expansion is the content of Section~\ref{sec:stae}, the proof of the adiabatic theorem itself is given in Section~\ref{sec:proofAdiabatic}. The   asymptotic expansion of the NEASS is stated in Proposition~\ref{proposition:expand} and proved in  Section~\ref{proposition:expand:proof}.
All statements of Section~\ref{sec:NEASS}    follow as corollaries of the more general space-time adiabatic  theory of Section~\ref{sec:spacetime}.
In Appendix~\ref{app:Vlemma} we prove   Lemma~\ref{lemma:Vcomm}  about commutators with  slowly varying potentials, while in the Appendices~\ref{app:tech}   we collect without proofs a number of technical lemmas  from other sources that need to be slightly adapted. Finally, Appendix~\ref{app:qli}  briefly discusses the local inverse of the Liouvillian 
and how to extend its mapping properties    to slowly varying potentials.

\medskip
 
\noindent {\bf Acknowledgements:}   I am grateful to Giovanna Marcelli, Domenico Monaco, and Gianluca Panati for their involvement in a closely related joint project.  
 I would like to thank Horia Cornean, Vojkan Jaksic, J\"urg Fr\"ohlich, and Marcello Porta for  very valuable discussions and  comments.
This work was  supported by the German Science Foundation within the Research Training Group 1838.

\section{The mathematical framework}\label{sec:frame}

In this section we explain the precise  setup necessary to formulate our main results. In a nutshell, we consider systems of interacting fermions on a subset $\Lambda\subset\Z^d$  of linear size $M$ of the $d$-dimensional square lattice $\Z^d$. In some directions $\Lambda$ can be closed in order to allow for cylinder or torus geometries. The Hamiltonian generating the dynamics is a quasi-local operator, that is, roughly speaking,  an extensive  sum of local operators. As we aim at statements that hold uniformly in the system size $M$, we consider actually families of operators indexed by $\Lambda$. 
This requires a certain amount of technical definitions, in particular, one needs norms that control families of quasi-local operators.
Since   the particle number depends on the size of the system, it is most convenient to work on Fock space. Many of the following concepts are standard and only slightly adapted  from \cite{BMNS,NSY,BDF}.

\subsection{The lattice and the Hilbert space}

   \begin{minipage}{7cm}  
\noindent Let $\Gamma = \Z^d$ be the infinite square lattice and  
$\Lambda =\Lambda(M) := \{-\frac{M}{2}+1,\ldots, \frac{M}{2}\}^d\subset \Gamma$  
the centred box of size $M$, with $M\in 2\N$.  
For many applications, in particular those   concerning currents, it is useful to consider $\Lambda$ being closed in some directions, say  in the  first $d_{\rm c}$  directions, $0\leq d_{\rm c}\leq d$. In particular, for $d_{\rm c}=d$ this means   that $\Lambda$ has a discrete torus geometry and for $d_{\rm c}=0$ it is a $d$-dimensional discrete cube. In order to define the corresponding metric on $\Lambda$, let $a\Mminus b$ be the representative of $[a-b]\in\Z/M$ in $\{-\frac{M}{2}+1,\ldots, \frac{M}{2}\}$ and define
 \end{minipage} \hspace{10mm}  
 \begin{minipage}{4cm}
 \hspace{5mm}
\includegraphics[height=3.5cm]{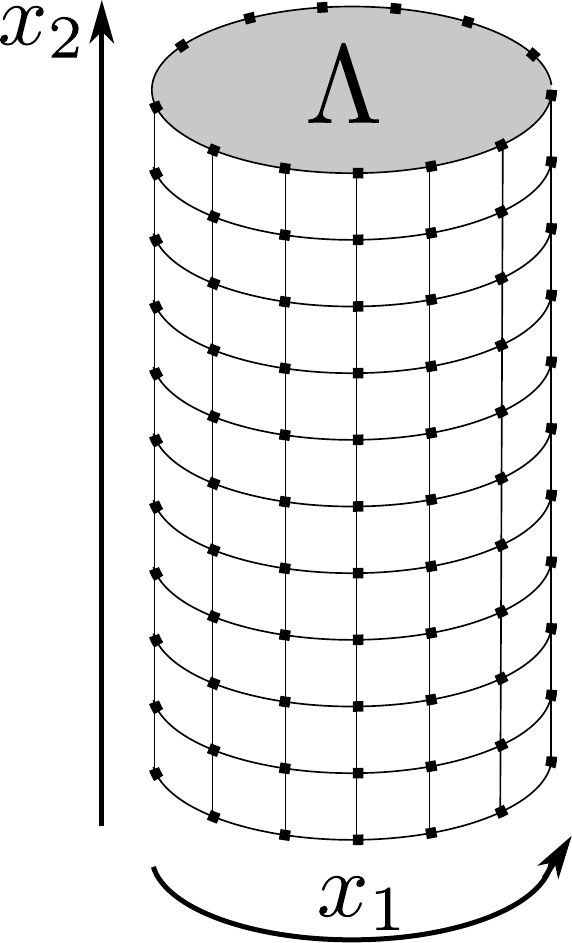} 
\vspace{2mm}
 
 \small Figure 2.1: Here $d=2$, and $\Lambda$ is closed in the $1$-direction and open in the $2$-direction.
\end{minipage}
\[
\lminus : \Lambda\times \Lambda \to\Gamma \,,\quad (x,y) \mapsto x\lminus y \quad\mbox{with}\quad (x\lminus y)_j =\left\{\begin{array}{cl}
x_j \Mminus y_j  & \mbox{ if $j\leq d_{\rm c}$}\\
x_j - y_j & \mbox{ if $j>d_{\rm c}$}\,,
\end{array}\right.
\]
 the difference vector of two points in $\Lambda$. With $d:\Gamma\times\Gamma \to \N_0$ denoting  the $\ell^1$-distance on $\Gamma$,
  the $\ell^1$-distance on the ``cylinder'' $\Lambda$ with the first $d_{\rm c}$ directions closed is
\[
d^\Lambda: \Lambda\times \Lambda \to \N_0\,,\quad d^\Lambda(x,y) := d( 0, y\lminus x)\,.
\]

Let the one-particle Hilbert space be $\mathfrak{h}_\Lambda := \ell^2(\Lambda, \C^\dimfib)$, 
$\dimfib\in\N$, where $\C^\dimfib$ describes   spin and the internal structure of the unit cell. 
The corresponding $N$-particle Hilbert space is its $N$-fold anti-symmetric tensor product   $\Hi_{\Lambda,N} := \bigwedge_{j=1}^N \mathfrak{h}_\Lambda$, 
and the fermionic Fock space is   $\mathfrak{F}_\Lambda := \bigoplus_{N=0}^{\dimfib M^d} \Hi_{\Lambda,N}$, 
where $\Hi_{\Lambda,0} := \C$. 
All these Hilbert spaces   are finite-dimensional and thus all linear operators on them  are bounded.
Let $a_{i,x}$ and $a_{i,x}^*$, $i=1,\ldots,\dimfib$, 
$x\in\Gamma$, be the standard fermionic 
annihilation and creation operators satisfying 
the canonical anti-commutation relations
\[
\{ a_{i,x}, a_{j,y}^* \} = \delta_{i,j} \delta_{x,y} {\bf 1}_{\mathfrak{F}_\Lambda}\quad\mbox{and}\quad \{ a_{i,x}, a_{j,y}  \} = 0 =  \{ a_{i,x}^*, a_{j,y}^*  \}\,,
\]
where $\{A,B\} = AB +BA$. 
For a   subset $X\subset \Lambda$ we denote by 
$\mathcal{A}_X\subset \mathcal{L}(\mathfrak{F}_\Lambda)$ 
the algebra of operators generated by the set
$\{   a_{i,x}, a_{i,x}^*\,|\, x\in X\,, i=1,\ldots, \dimfib\}$. 
Those elements of $\mathcal{A}_X$ commuting with the number operator 
\[
\mathfrak{N}_X := \sum_{x\in X} a_x^*a_x := \sum_{x\in X} \sum_{j=1}^\dimfib a_{j,x}^*a_{j,x}
\]
form a sub-algebra $\mathcal{A}_X^\mathfrak{N}$ of $\mathcal{A}_X$ contained in the sub-algebra $\mathcal{A}_X^+$ of even elements, i.e.\ $\mathcal{A}_X^\mathfrak{N}\subset  \mathcal{A}_X^+\subset  \mathcal{A}_X $. We will use the vector notation  $a_x = (a_{1,x},\ldots, a_{\dimfib,x})$   without further notice in the following.

\subsection{Interactions and associated operator-families}

An  interaction $\Phi= \{\Phi^{\epsi,\Lambda}\}_{\epsi\in(0,1],\,\Lambda=\Lambda(M), \, M \in 2\N}$ is a family of maps  
\[
\Phi^{\epsi,\Lambda}:  \{ X\subset \Lambda\} \to \bigcup_{X \subset\Lambda} \mathcal{A}_X^\mathfrak{N} \,,\quad X\mapsto\Phi^{\epsi,\Lambda}(X)\in \mathcal{A}_X^\mathfrak{N}\,,
\]
from subsets of $\Lambda$ into   the set of   operators commuting with the number operator $\mathfrak{N}_X$.
The   operator-family  $A  = \{ A^{\epsi,\Lambda}\}_{\epsi\in(0,1],\,\Lambda}$ 
associated with the   interaction  $\Phi$ is the family of   operators
\begin{equation}\label{def:ham}
A^{\epsi,\Lambda} \equiv A^{\epsi,\Lambda}(\Phi) :=  \sum_{X \subset \Lambda} \Phi^{\epsi,\Lambda}(X)  \in \mathcal{A}_\Lambda^\mathfrak{N}\,.
\end{equation}
In the case that an interaction or the associated operator-family does not depend 
on the parameter~$\epsi$, we will drop the superscript $\epsi$ in the notation.
Note also that ``interaction'' is used here as a mathematical term for the above kind of object and should not be 
confused with the physics notion of interaction.

In order to turn the vector space of interactions into a normed space, it is useful to introduce 
the following functions that will serve to control the range of an interaction (cf.\ e.g.\ \cite{NSY} and references therein). Let 
\[
F (r) := \frac{1}{(1+r)^{d+1}}   
\qquad\mbox{and}\qquad 
F_\zeta(r) := \frac{\zeta(r)}{(1+r)^{d+1}}\,,
\]
where 
\begin{align*}
\zeta\in \mathcal{S} & := \{ \zeta:[0,\infty)\to (0,\infty)\,|\, \mbox{$\zeta$ is bounded, non-increasing, satisfies }\\
&\quad\quad\zeta(r+s) \geq \zeta(r)\zeta(s)\;\mbox{ for all } r,s\in[0,\infty)\mbox{ and } \\
&\quad\quad \sup\{ r^n\zeta(r)\,|\, r\in [0,\infty)\} <\infty \mbox{ for all } n\in\N\}\,.
\end{align*}
For any $\zeta \in \mathcal{S}$ and $n\in \N_0$, a corresponding norm on the vector space of  interactions is   defined by
 \[
\|\Phi\|_{\zeta,n } :=   \sup_{\epsi\in(0,1]}\sup_{\Lambda} \sup_{x,y\in\Lambda} \sum_{\substack{X \subset \Lambda:\\ \{x,y\}\subset X}} \mbox{$\Lambda$-diam}(X)^n \frac{\|\Phi^{\epsi,\Lambda}(X)\|}{F_\zeta(d^\Lambda (x,y))}  \,.
\]
 Here $\mbox{$\Lambda$-diam}(X)$ denotes the diameter of the set $X\subset\Lambda$ with respect to the metric $d^\Lambda$.
 The prime example for a function $\zeta\in\mathcal{S}$ is  $\zeta(r) = \E^{-ar}$ for some $a>0$. For this specific choice of $\zeta$ we write $F_a$ and $\|\Phi\|_{a,n}$ for the corresponding norm. 
However, for technical reasons the use of the more general decay functions $\zeta$ in $\mathcal{S}$ seems unavoidable, see also the remark after Lemma~\ref{lemma:I2}.
 \begin{SCfigure}[1.15]
   \includegraphics[width=0.33\textwidth]{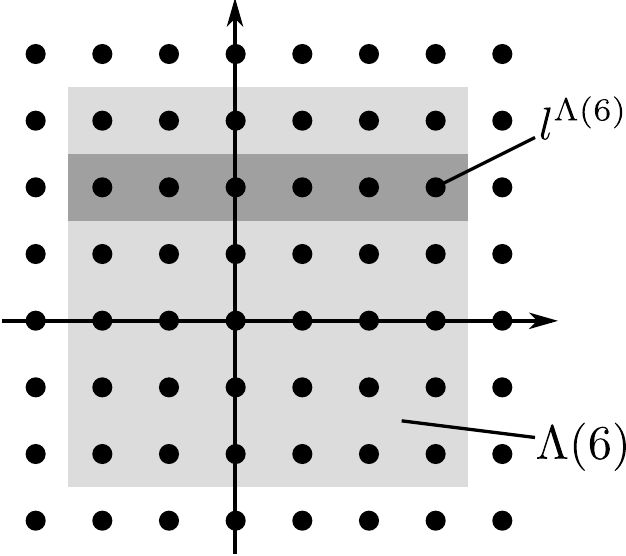}
    \caption{\small The light shaded region is the cube $\Lambda(6)$. The darker shaded region is the hyperplane defined by the localization vector $L$ with $\ell =(0,1)$ corresponding to localisation   in the $2$-direction around the point $l^{\Lambda(6)}$. An interaction with finite $\|\cdot\|_{\zeta,n,L}$-norm is then localized near this hyperplane. } 
    \label{fig:1}
\end{SCfigure}

It is important for applications to consider also interactions 
that are localised in certain directions around 
certain locations. To this end we introduce the 
space of localisation vectors 
\[
\loc :=  \{0,1\}^d \times {\textstyle \prod}_{M=2}^\infty \,\Lambda(M)\,.
\]
Note that
 $L=: ( \ell, l^{\Lambda(2)},  l^{\Lambda(4)}, \ldots)\in \loc$ 
 defines for each $\Lambda$ a $(d-|\ell|)$-dimensional hyperplane 
 through the point $l^\Lambda\in\Lambda$ which is parallel to the 
 one given by $\{x_j=0 \text{ if } \ell_j=1\}$.
 Here $|\ell|:=|\{ \ell_j =1\}|$ is the number   of   constrained directions. 
 The distance of a point $x\in\Lambda$ to this hyperplane is
 \begin{equation} \label{def:distxL}
 {\rm dist}(x,L) := \sum_{j=1}^d   | (x\lminus l^\Lambda)_j| \,\ell_j
 \end{equation}
 and we define for $\alpha>0$ 
a  new ``distance'' on $\Lambda$ by 
 \[
 d^\Lambda_{L_\alpha}(x,y) := d^\Lambda(x,y) +  \alpha \left( {\rm dist}(x,L) + {\rm dist}(y,L) \right)\,.
 \] 
 Note that $d^\Lambda_{L_\alpha}$ is no longer a metric  on $\Lambda$ 
 but obviously still satisfies the triangle inequality. Moreover, if $L\in \loc$ is trivial, i.e.\ $|\ell|=0$, then 
 $ d^\Lambda_{L_\alpha} =  d^\Lambda $. 
 The corresponding norms are denoted by 
  \[
\|\Phi\|_{\zeta,n,L_\alpha } :=   \sup_{\epsi\in(0,1]}\sup_{\Lambda} \sup_{x,y\in\Lambda} \sum_{\substack{X \subset \Lambda:\\ \{x,y\}\subset X}} \mbox{$\Lambda$-diam}(X)^n \frac{\|\Phi^{\epsi,\Lambda}(X)\|}{F_\zeta(d^\Lambda_{L_\alpha} (x,y))}  \,.
\]
In the following we will either have $\alpha=1$ or $\alpha=\epsi$.
 
An operator-family $A$ with interaction $\Phi_A$ 
such that $\|\Phi_A\|_{\zeta,0,L_\alpha}<\infty$ 
  for some $\zeta\in\mathcal{S}$ is called 
\emph{quasi-local and $L_\alpha$-localised}.  
One crucial property of quasi-local $L_\alpha$-localised 
operator-families is that the norm of the finite-size operator $A^{\epsi,\Lambda}$
grows at most as the volume  of its ``support'' (cf.\ Lemma~\ref{lemma:bound})
\[
\|A^{\epsi,\Lambda}\|\leq C\, \|\Phi_A\|_{\zeta,0,L_\alpha} \, \alpha^{-|\ell|} \,M^{d-|\ell|}\,.
\]
Let $\mathcal{B}_{\zeta,n,L_\alpha}$ be the 
Banach space of   interactions with 
finite $\|\cdot\|_{\zeta,n,L_\alpha}$-norm, and put
\begin{eqnarray*}
&& \mathcal{B}_{\mathcal{S},n,L_\alpha} :=  \bigcup_{\zeta\in\mathcal{S}}  \mathcal{B}_{\zeta,n,L_\alpha}\,, \qquad \mathcal{B}_{\mathcal{E},n,L_\alpha} := \bigcup_{a>0}  \mathcal{B}_{a,n,L_\alpha} \,,\\&&
  \mathcal{B}_{\mathcal{S},\infty,L_\alpha} := \bigcap_{n\in\N_0}  \mathcal{B}_{\mathcal{S},n,L_\alpha} \,,\; \quad \mathcal{B}_{\mathcal{E},\infty,L_\alpha} := \bigcap_{n\in\N_0}  \mathcal{B}_{\mathcal{E},n,L_\alpha}\,.
\end{eqnarray*}
Note that $\Phi \in \mathcal{B}_{\mathcal{S},\infty,L_\alpha} $ 
merely means that there exists a sequence $\zeta_n\in\mathcal{S}$ such that
$\Phi\in\mathcal{B}_{\zeta_n,n,L_\alpha}$ for all $n\in\N_0$.
The corresponding spaces of   operator-families are denoted by 
$\mathcal{L}_{\zeta,n,L_\alpha}$, $ \mathcal{L}_{\mathcal{E},n,L_\alpha}$, 
$\mathcal{L}_{\mathcal{E},\infty,L_\alpha}$, $ \mathcal{L}_{\mathcal{S},n,L_\alpha}$, 
and $\mathcal{L}_{\mathcal{S},\infty,L_\alpha}$ respectively: that is, an  
operator-family $A$ belongs to $\mathcal{L}_{\zeta,n,L_\alpha}$ if it can 
be written in the form \eqref{def:ham} with an  
 interaction in~$\mathcal{B}_{\zeta,n,L_\alpha}$, and similarly for the other spaces. 
  Lemma~A.1 in \cite{MT} shows that the spaces $ \mathcal{B}_{\mathcal{S},n,L_\alpha} $
   and therefore also $ \mathcal{B}_{\mathcal{S},\infty,L_\alpha} $ are indeed  vector spaces.
   Note that they are   not  {algebras} 
of operators, i.e.\ the product of two quasi-local $L_\alpha$-localised 
operators  need neither be quasi-local nor $L_\alpha$-localised, but 
they are closed under taking  {commutators},
see Lemmas~\ref{lemma:commutator1} and \ref{lemma:ads} 
in Appendix~\ref{app:tech}. When we don't write the
 index $L_\alpha$, this means that $L=0 := (\vec{0},0,0,\ldots) $ and the 
 interaction (respectively the operator-family) is quasi-local but not localised 
 in any direction.

\subsection{Slowly varying potentials}

Roughly speaking, a slowly varying potential is a function $v$ on $\Lambda$ 
with Lipschitz-constant of order $\epsi$. 

\begin{definition}[\bf Slowly varying potentials]
A  slowly varying potential is a family of  functions   $v = \{v^{\epsi, \Lambda} :\Lambda \to \R\}_{  \epsi\in(0,1],\,\Lambda}$  such that  
\[
   C_v := \sup_{\epsi\in(0,1]}\sup_\Lambda \sup_{x,y\in \Lambda}\frac{ |  v^{\epsi,\Lambda}( x) -   v^{\epsi,\Lambda}(y) |}{\epsi\cdot d^\Lambda(x,y)}  <\infty\,.
\]
The space of slowly varying potentials is denoted by $\SVP$.
\end{definition}
Let us give some typical examples. If $\Lambda$ is open in the $j$-direction, then for any $v_1\in C^1(\R,\R)$ (not necessarily bounded!) with $\|v_1'\|_\infty<\infty$ the functions
\begin{eqnarray*}
v_1^{\epsi,\Lambda}:\Lambda\to \R\,,&& x\mapsto v^{\epsi,\Lambda}_1(x) := v_1(\epsi x_j)\\
\tilde v_1^{\epsi,\Lambda}:\Lambda\to \R\,,&& x\mapsto \tilde v^{\epsi,\Lambda}_1(x) := \epsi v_1(  x_j)
\end{eqnarray*}
 both satisfy, by the mean-value theorem,   for any $x,y\in \Lambda$
 \[
  |v_1^{\epsi,\Lambda}( x) -   v_1^{\epsi,\Lambda}(y)|   \leq  \epsi\cdot d^\Lambda(x,y) \, \|v_1'\|_\infty\quad\mbox{and}\quad 
 | \tilde v_1^{\epsi,\Lambda}( x) -   \tilde v_1^{\epsi,\Lambda}(y)|   \leq   \epsi\cdot d^\Lambda(x,y) \,\|v_1'\|_\infty
  \,.
\]
Hence, $v_1,\tilde v_1 \in \SVP$ with $C_{v_1}= C_{\tilde v_1}= \|v_1'\|_\infty$ in this case. In particular, $v_1(x) = x_j$ would be a viable option.
In the case that $\Lambda$ is closed in  the $j$-direction, any function  $v_2: (-1/2,1/2] \to \R$ that is $C^1$ with periodic boundary conditions defines, by the same reasoning, 
 a slowly varying potential 
\[
v_2^{\epsi,\Lambda} (x) :=  \epsi M v_2\left(\frac{x_j}{M}\right)
\quad \mbox{with}\quad C_{v_2} =   \|    v_2' \|_\infty\,.
\]
With a slowly varying potential $v^{\epsi,\Lambda}$ we associate a corresponding operator-family $V_v$   defined by  
\[
V_v^{\epsi,\Lambda} := \sum_{x\in\Lambda} v^{\epsi,\Lambda}(x) \,a^*_xa_x\,.
\]
The key property of a slowly varying potential is that its commutator with any quasi-local operator-family  is a quasi-local operator-family of order $\epsi$, as stated more precisely in the following lemma.
\begin{lemma}\label{lemma:Vcomm}
Let $V_v$ be the operator-family of a slowly varying potential $v\in\SVP$ and $A\in \mathcal{L}_{\mathcal{S},\infty,L_\alpha}$, i.e.\ $A\in \mathcal{L}_{\zeta_k,k+d+1,L_\alpha}$ for some sequence $\zeta_k$ in $\mathcal{S}$.  Then there exists an operator-family $A_{v}\in \mathcal{L} _{\mathcal{S},\infty,L_\alpha}$
with interaction $\Phi_{A_{v}}\in \mathcal{B} _{\mathcal{S},\infty,L_\alpha}$ satisfying 
\[
 \|\Phi_{A_{v}}\|_{\zeta_k,k,L_\alpha} \leq \tfrac{\dimfib}{2} \,C_v \,\| \Phi^\Lambda_{A}\|_{\zeta_k,k+d+1,L_\alpha}
\]
such that 
\[
[A, V_v] = \epsi A_{v}\,.
\]
\end{lemma}
\noindent The proof of Lemma~\ref{lemma:Vcomm} is given in Appendix~\ref{app:Vlemma}.

\section{Non-equilibrium almost-stationary states} \label{sec:NEASS}

We have now all the tools to formulate our main result about non-equilibrium almost-stationary states (NEASS) for time-independent  Hamiltonians   of the form 
\[
H = H_0 + V_v + \epsi H_1\,.
\]
The results of this section will   follow as corollaries of the space-time adiabatic theorem, Theorem~\ref{theorem:Adi} of Section~\ref{sec:spacetime} and the underlying space-time adiabatic expansion, Proposition~\ref{proposition:Adi}.  However, we consider these special cases conceptually most important and it might be difficult  to read  them directly off  the rather technical general statement. Moreover, they are at the basis of the proof of linear response theory given in Section~\ref{sec:LinearResponse}. 

\bigskip
\noindent {\bf (A1)  Assumptions on $H_0$.}\\ {\em
Let $H_0 \in  \mathcal{L}_{a,n}$   for all $n\in\N_0$ and some $a>0$ such that all $H_0^\Lambda$ are self-adjoint. We assume that there exists $M_0\in\N$ such that for all $M\geq M_0$ and corresponding $\Lambda=\Lambda(M)$  the ground state $E_*^\Lambda $ of the operator $H_0^\Lambda$, with associated spectral projection $P_*^\Lambda $, has the following properties:   
The degeneracy $\dege^\Lambda $ of $E_*^\Lambda $ and  the spectral gap are uniform in the system size, i.e.\  there exist $\dege\in\N$  and $g>0$  such that
$\dege^\Lambda \leq \dege$ and $\mathrm{dist}(E_*^\Lambda , \sigma(H_0^\Lambda )\setminus E_*^\Lambda )\geq g >0$ for all $M\geq M_0$.
}
\medskip

A typical  example of a physically relevant  Hamiltonian $H_0$ to which our results apply is the family of operators
\begin{equation}\label{example:HTPHW}
H^\Lambda_{T\phi W\mu}   =  \sum_{ (x,y)\in \Lambda^2}  \hspace{-8pt}  a^*_x \,T( x\lminus y) \,a_y+ \sum_{x\in\Lambda}  a^*_x\phi( x)a_x  
  + \hspace{-4pt} \sum_{\{x,y\}\subset \Lambda } \hspace{-8pt} a^*_xa_x \,W( d^\Lambda(x,y))\,a^*_ya_y - \mu \,\mathfrak{N}_\Lambda\,.
\end{equation}
For example, if the kinetic term $T :\Gamma \to \mathcal{L}(\C^\dimfib)$ is a compactly supported function with $T( -x) = T(  x)^*$, the potential term $\phi :\Gamma \to \mathcal{L}(\C^\dimfib)$ is a bounded function taking values in the self-adjoint matrices,   the two-body interaction $W :\N_0\to 
\mathcal{L}(\C^\dimfib)$ is compactly supported and also takes values in the self-adjoint matrices, then $H_{T\phi W\mu}  \in \mathcal{L}_{a,\infty}$ for any $a>0$. 
For non-interacting systems, i.e.\  $W\equiv0$, on a torus and $V$ sufficiently small, the gap condition  can be checked rather directly and one typically finds it to be satisfied for   values of the chemical potential $\mu\in\R$ lying in specific intervals.  It was recently shown in \cite{H,DS} that for sufficiently small  $W\not=0$ the gap remains open. 

\bigskip

\noindent {\bf (A2)$_{L^H,\gamma}$  Assumptions on the perturbations.} \\{\em
Let $H_1 \in\mathcal{L}_{\mathcal{S},\infty,L^H_1}$ be self-adjoint  and let $v\in \SVP$ be a  slowly varying potential and   $V_v$ the corresponding operator-family.  
If $L^H$ is nontrivial we assume that $V_v$ is $L^H_{\epsi^\gamma}$-localised for $\gamma\in\{0,1\}$ in the sense
that
 \[ [H_0 , \tfrac{1}{\epsi}V_v ]\in \mathcal{L}_{\mathcal{S},\infty,  L^H_{\epsi^\gamma}}\,.
 \]
}
\medskip

\begin{remark}\rm
\balp
\item The parameter $\gamma\in \{0,1\}$ will appear in the statements of our results for the following reason. While a perturbation of the form $  H_1$ can be localised in a fixed  neighbour\-hood of a $(d-|\ell_H|)$-dimensional hyperplane and thus on a volume of order $ M^{d-|\ell_H|}$, a slowly varying potential is typically only localised in an $\epsi^{-1}$-neighbour\-hood of such a $(d-|\ell_H|)$-dimensional hyperplane\footnote{Think for example of a smooth step function along a codimension one hyperplane with finite step size but order $\epsi$ slope.}, and thus on a volume of order $\epsi^{-|\ell_H|} M^{d-|\ell_H|}$.
Hence, the proper normalisation   in the statements will then contain a factor $\epsi^{-|\ell_H|\gamma} M^{d-|\ell_H|}$.

\item
It is not known whether perturbing a generic gapped $H_0$ by a small local perturbation $\epsi H_1$ leaves  the spectral gap open for  $\epsi$   sufficiently small, see e.g.\ the discussion in Section~1.5 in \cite{BBDF}.  However,   perturbing by a slowly varying potential $V_v$ such that $\|v^\Lambda\|_\infty\sim M$  clearly  closes the gap of $H_0$ for all values of $\epsi$.
\ealp
\end{remark}

In a nutshell, the following theorem about non-equilibrium almost-stationary states  says that for any $n\in\N$ there exists a state $\Pi_n$, that is obtained from the ground state $P_*$ of $H_0$ 
 by a  unitary transformation with small quasi-local generator, such that $[\Pi_n , H] =\Or(\epsi^{n+1})$. As a consequence, $\Pi_n$ is almost-invariant under the dynamics generated by $H$. In cases where the perturbation is so small that the gap of $H_0$ remains open, $\Pi_n$ is $\epsi^{n+1}$-close to the ground state projection of $H$ in the sense that the Taylor polynomials of both operators agree up to order $\epsi^n$. In general, however, $\Pi_n$ can     differ  greatly from the ground state or any other thermal equilibrium state of~$H$. This is why we call $\Pi_n$ a non-equilibrium almost-stationary state.

\begin{theorem}[\bf Non-equilibrium almost-stationary states]\label{SpaceAdiabaticThm}
 Let the Hamiltonian $H  =H_0 + V_v + \epsi H_1$ satisfy (A1) and (A2)$_{L^H,\gamma}$ for some $L^H\in \loc$ and $\gamma\in\{0,1\}$.  Then there is a
sequence of  self-adjoint operator-families $(A_\mu)_{\mu\in\N}$ with $  A_\mu\in \mathcal{L}_{\mathcal{S},\infty,L^H_{\epsi^\gamma} }$   for all $\mu\in\N$, such that for any $n\in\N$ it holds
that the projector 
\[
  \Pi^{\epsi,\Lambda}_n  := \E^{\I \epsi  S^{\epsi,\Lambda}_n } \,  P_*^{\Lambda}  \E^{- \I \epsi  S^{\epsi,\Lambda}_n } \quad\mbox{with}\quad  
 S^{\epsi,\Lambda}_n  := \sum_{\mu = 1}^{n} \epsi^{\mu-1}  A^{\epsi,\Lambda}_\mu  
   \]
satisfies 
\begin{equation}\label{NeassExpand}
 [  \Pi^{\epsi,\Lambda}_n , H^{\epsi,\Lambda}]  = \epsi^{n+1} \, [  \Pi^{\epsi,\Lambda}_n ,R^{\epsi,\Lambda}_n]
\end{equation}
for some $R _n \in  \mathcal{L}_{\mathcal{S},\infty,L^H_{\epsi^\gamma} }$.

The state $\Pi^{\epsi,\Lambda}_n$ is almost-stationary  for  the dynamics generated by $H^{\epsi,\Lambda}$ in the following sense:
Let $\rho^{\epsi,\Lambda}(t)$ be the solution of the Schr\"odinger equation
 \[ 
\I \tfrac{\D}{\D t}  \,  \rho^{\epsi , \Lambda}(t) = [ H^{\epsi,\Lambda}  ,  \rho^{\epsi, \Lambda}(t)] \qquad 
  \mbox{with}\qquad  \rho^{\epsi, \Lambda}(0) = \Pi_n^{\epsi,  \Lambda} \,.
  \] 
    Then for any $\zeta\in \mathcal{S}$, $L\in\loc$ with $\ell\cdot\ell_H=0$, and $m\geq 1$ there exists a constant $C$ such that for any $B\in \mathcal{L}_{\zeta,2d,L_1}$   
  \begin{eqnarray}\label{AdiState00}\lefteqn{\hspace{-2cm}
\sup_{\Lambda(M):M\geq M_0} \;  \tfrac{\epsi^{|\ell_H| \gamma}}{M^{d -|\ell|-|\ell_H|}} \left|  \tr  (\rho^{\epsi, \Lambda}(t)  B^\Lambda )- \tr  (  \Pi_n^{\epsi, \Lambda}  B^\Lambda)\right|}\nonumber \\& &\hspace{2cm} \leq\;C \, {\epsi^{n +1 }}   \, (1+|\epsi^m t|^{d+1}) \,\|\Phi_B\|_{\zeta,2d,L_1}
\,.
\end{eqnarray}
\end{theorem}

\noindent 
While \eqref{NeassExpand} is an immediate consequence of Proposition~\ref{proposition:Adi} for $\delta=0$, 
\eqref{AdiState00} is not strictly speaking a corollary of the results of Section~\ref{sec:spacetime}. But it can easily  be  concluded by combining  Proposition~\ref{proposition:expand1} with a simple Duhamel argument for $\Pi_{\tilde n}^{\epsi, \Lambda} $ with $\tilde n = n+m(d+1)$ and $\delta=\epsi^{m-1}$ as in the proof of Theorem~\ref{theorem:Adi}.

\medskip
\begin{remark}\rm 
\balp
 \item The trace in \eqref{AdiState00} is normalised by the volume of the region in~$\Lambda$ where the perturbation ``acts'' 
 and the observable $B$ ``tests''. This region is the intersection of the neighbourhoods of two transversal hyperplanes, 
 one of co-dimension $ |\ell_H|$ and thickness of order $\epsi^{-1}$ (resp.\ of order one if $\gamma=0$), and one of co-dimension $ |\ell|$ and  thickness of order one.
 Hence, the relevant volume is of order  $ \frac{M^{d -|\ell|-|\ell_H|}}{\epsi^{|\ell_H|\gamma}} $.
 If the perturbation acts everywhere ($|\ell_H|=0$) and the observable $B$ tests everywhere ($|\ell|=0$), then the trace in \eqref{AdiState00} is the usual trace per unit volume.
 \item Note that 
  the first $n$ terms in the asymptotic expansion of the NEASS $\Pi^{\epsi,\Lambda}_n$ are
 uniquely determined by the requirements that $(\Pi^{\epsi,\Lambda}_n)^2-\Pi^{\epsi,\Lambda}_n=\Or(\epsi^{n+1}) $ and $[H^{\epsi,\Lambda}, \Pi^{\epsi,\Lambda}_n]=\Or(\epsi^{n+1})$,   cf.\ e.g.\ \cite{PST2}. Hence,  all terms in this expansion could be equally well obtained by just applying standard regular perturbation theory, e.g.\ \cite{K}. However, the unitary $ \E^{\I \epsi  S^{\epsi,\Lambda}_n }$ is not uniquely determined and one key feature of the above result is that the latter can be chosen quasi-local. Otherwise  the almost invariance in \eqref{AdiState00} uniformly in the volume could not be concluded.

   \item The NEASS $\Pi^{\epsi,\Lambda}_n$ agrees with the state obtained by the quasi-adiabatic 
   evolution  \cite{HW,BMNS}  up to terms of order $\epsi^{n+1}$ in the following sense:  The quasi-adiabatic evolution of $Q^{0,\Lambda} := P_*^\Lambda$ is the solution to the evolution equation 
   \begin{equation}\label{SpecFlow}
\I\tfrac{\D}{\D\epsi} Q^{\epsi,\Lambda} = [ \mathcal{I}_{H^{\epsi,\Lambda}}( \tfrac{\D}{\D\epsi}H^{\epsi,\Lambda}), Q^{\epsi,\Lambda}]\quad\mbox{with}\quad Q^{0,\Lambda} = P_*^{ \Lambda}\,.
\end{equation}
   Here $ \mathcal{I}_{H^{\epsi,\Lambda}}$ is the quasi-local inverse of the Liouvillian discussed in Appendix~\ref{app:qli}. It is straightforward to check that if  $H^{\epsi,\Lambda}$    has a gapped ground state $P_*^{\epsi,\Lambda}$ uniformly for all $\epsi\in[0,\epsi_0]$, then the quasi-adiabatic evolution $Q^{\epsi,\Lambda} $ of $P_*^\Lambda$ actually agrees with $P_*^{\epsi,\Lambda}$, i.e.\ $Q^{\epsi,\Lambda}=P_*^{\epsi,\Lambda}$ for all $\epsi\in[0,\epsi_0]$.
In particular,    $Q^{\epsi,\Lambda} $ and $P_*^{\epsi,\Lambda}$ have the same Taylor expansion in powers of $\epsi$. But also  the NEASS $\Pi^{\epsi,\Lambda}_n$ and the eigenprojection $P_*^{\epsi,\Lambda}$ have the same Taylor expansion in powers of $\epsi$ up to order $n$, as remarked under item (b).

 \ealp
\label{remark:NEASS}
\end{remark}

For applications it is   essential to have an explicit expansion of expectation values in the NEASS   in powers of $\epsi$ with coefficients given by expectations in the unperturbed  ground state, the linear term being typically most important. The following statement is a special case of Proposition~\ref{proposition:expand} about the expansion of the time-dependent NEASS.
For better readability we will   use the notation 
\[
\langle B  \rangle_{\rho^\Lambda} := \tr(\rho^\Lambda \,B^\Lambda)
\]
for the expectation value of an observable $B^\Lambda$  in a  state  $\rho^\Lambda$.

\begin{proposition}[\bf Asymptotic expansion of the NEASS]\label{proposition:expand1}
Under the assumptions of Theorem~\ref{SpaceAdiabaticThm}
there exist linear maps $\mathcal{K}^{\epsi,\Lambda}_j: \mathcal{A}_\Lambda^\mathfrak{N} \to  \mathcal{A}_\Lambda^\mathfrak{N}$, $j\in\N$, given by nested commutators with operator-families in $\mathcal{L}_{\mathcal{S},\infty,L^H_{\epsi^\gamma} }$, such that for any $n,m\in\N_0$ with $n\geq m$, any $\zeta\in \mathcal{S}$, and any $L\in\loc$ with $\ell\cdot\ell_H=0$ there is a constant $C$ such that
  for any $B\in \mathcal{L}_{\zeta,m+1,L_1}$ it holds that 
\[
\sup_{\Lambda(M):M\geq M_0}\,  \tfrac{\epsi^{|\ell_H|\gamma}}{M^{d-|\ell|-|\ell_H|}} \left|  \, \left\langle B  \right\rangle_{\Pi^{\epsi,\Lambda}_n} -   \sum_{j=0}^m \epsi^j\,
\left\langle \mathcal{K}^{\epsi }_j [B ]  \right\rangle_{P_*^\Lambda}
\right| \;=\; C\,\epsi^{m+1}\, \|\Phi_B\|_{\zeta,m+1,L_1} \,,
\]
with
\[
\sup_{\Lambda(M):M\geq M_0}\; \sup_{\epsi\in(0,1]}  \tfrac{\epsi^{|\ell_H|\gamma}}{M^{d-|\ell|-|\ell_H|}} \,\left\langle \mathcal{K}^{\epsi }_j [B ]  \right\rangle_{P_*^\Lambda} <\infty
\]
for all $j\in \N$.  
The first terms in the expansion are given by
 \[
\mathcal{K}_0^{\epsi,\Lambda} = {\rm Id} \,,\quad  \mathcal{K}_1^{\epsi,\Lambda} [\cdot] =- \I\, [A_1^{\epsi,\Lambda},\,\cdot\,] \,, \quad\mbox{and}\quad  
\mathcal{K}_2^{\epsi,\Lambda} [\cdot] =- \I\, [A_2^{\epsi,\Lambda},\,\cdot\,]
- \tfrac{1}{2}[A_1^{\epsi,\Lambda},[A_1^{\epsi,\Lambda},\cdot ]]\,,
 \]
 where $A_1 \in\mathcal{L}_{\mathcal{S},\infty,L^H_{\epsi^\gamma} }$ and $A_2 \in\mathcal{L}_{\mathcal{S},\infty,L^H_{\epsi^\gamma} }$ were constructed in Theorem~\ref{SpaceAdiabaticThm}.
More expli\-citly, abbreviating $V^{ \Lambda}:= \tfrac{1}{\epsi}V^{\epsi,\Lambda}_v+H_1^{\epsi,\Lambda}$ it holds that
 \begin{equation}\label{K1}
 \left\langle \mathcal{K}^{\epsi }_1[B ]  \right\rangle_{P_*^\Lambda}=
 \left\langle \left[  \left[ R _0,V \right],B 
\right]\right\rangle_{P_*^\Lambda}
 \end{equation}
 and
  \begin{eqnarray}\label{K2}
 \left\langle \mathcal{K}^{\epsi }_2 [B]  \right\rangle_{P_*^\Lambda} &=&2{\rm Re} \,
  \left\langle   \left(   V R_0V R _0-V P_*  V  R _0R _0 \right)B
       \right\rangle_{P_*^\Lambda }
\nonumber \\
&& + 
 \left\langle   \left(V R _0B  R _0V -V R _0R _0V P_*B \right)
       \right\rangle_{P_*^\Lambda}
\;\,,\, 
 \end{eqnarray}
where 
 $R^\Lambda_0 := (H_0^\Lambda - E_*^\Lambda)^{-1} ( 1- P_*^\Lambda)$ denotes the reduced resolvent of $H_0$.
\end{proposition}

 Finally, again as a corollary of Proposition~\ref{proposition:Adi} and Theorem~\ref{theorem:Adi}, we note that  adia\-batically  switching on the perturbation  drives the ground state of the unperturbed Hamiltonian into the NEASS of the perturbed Hamiltonian up to small errors in the adiabatic parameter and independently of the precise form of the switching function.

 \begin{proposition}[\bf Adiabatic switching and the NEASS]\label{prop:switch}
  Let the Hamiltonian $H  =H_0 + V_v + \epsi H_1$ satisfy (A1) and (A2)$_{L_H,\gamma}$ for some $L_H\in \loc$ and $\gamma\in\{0,1\}$. Let   $f:\R \to \R$ be a smooth ``switching'' function with   
  $f(t) =0$ for $t\leq -1$ and $f(t) = 1$ for $t\geq  0$, and define  $H(t) := H_0 + f(t) (V_v + \epsi H_1)$. 
  Let $\rho^{\epsi , \Lambda,\eta,f}(t)$ be   the solution of the adiabatic time-dependent Schr\"odinger equation 
   \begin{equation}\label{Schroedinger2}
\I\,   \eta \tfrac{\D}{\D t}  \,  \rho^{\epsi , \Lambda,\eta,f}(t) = [ H^{\epsi,\Lambda}(t)  ,  \rho^{\epsi, \Lambda,\eta,f}(t)] 
  \end{equation}
  with adiabatic parameter $\eta\in(0,1]$  and initial datum $\rho^{\epsi, \Lambda,\eta,f}(t) = P_*^\Lambda$ for all $t\leq -1$.

Then
  for any $n>d$, $\zeta\in \mathcal{S}$ and $L\in\loc$ with $\ell\cdot\ell_H=0$ there exists a constant $C$ such that for any $B\in \mathcal{L}_{\zeta,2d,L_1}$  and  for all $t\geq 0$  
  \begin{align*} 
\sup_{\Lambda(M):M\geq M_0} \;  \tfrac{\epsi^{|\ell_H|\gamma}}{M^{d -|\ell|-|\ell_H|}} & \left| 
\left\langle B \right\rangle_{\rho^{\epsi, \Lambda,\eta,f}(t) }
- \left\langle B \right\rangle_{\Pi_n^{\epsi, \Lambda} }
\right|  \\& \leq \;\frac{\epsi^{n+1} + \eta^{n+1}}{\eta^{d+1}}  \,C \,(1+t^{d+1}) \,\|\Phi_B\|_{\zeta,2d,L_1} 
\,,
\end{align*}
where $\Pi_n^{\epsi, \Lambda} $ is the NEASS of $H^{\epsi,\Lambda}$ constructed in Theorem~\ref{SpaceAdiabaticThm}.
 \end{proposition}

\begin{remark}\label{timescalerem}\rm
Proposition~\ref{prop:switch} shows that, as long as the adiabatic parameter $\eta$ satisfies 
\[
1\gg \eta \gg \epsi^\frac{n+1}{d+1} \,,
\]
the initial ground state $P_*$ of $H_0$ evolves, up to a small error, into a NEASS $\Pi_n^{\epsi, \Lambda} $ that is independent of the form of the switching function $f$. 
Since $n\in\N$ can be chosen arbitrarily large, this means that whenever the adiabatic switching of the perturbation   occurs on a  time-scale of order $\epsi^{-m}$ with $m>0$,  the system will be driven into the same unique NEASS constructed in Theorem~\ref{SpaceAdiabaticThm}. 
Slower switching must be excluded, because, in general, the NEASS is an almost-invariant but not an invariant state for the instantaneous Hamiltonian. On time scales asymptotically larger than any inverse power of $\epsi$, the NEASS need not be  stable and might deteriorate because of tunnelling. Hence,  it is not surprising that the relevant time scale for the adiabatic switching process depends on the strength of the perturbation.
 \end{remark}

\section{Linear response theory}\label{sec:LinearResponse}
 
To    put our result on the justification of response theory, Theorem~\ref{response},  into proper context, 
  we   briefly recall the usual derivation of linear response formulas in the context of static perturbations. 
Assume that a system described by the Hamiltonian $H_0^\Lambda$ is initially  in its zero-temperature   equilibrium state~$P_*^\Lambda$, when a static perturbation $\epsi V^{\Lambda}$ is applied. To keep notation concise, in this section we will  again  write $V^{ \Lambda}:= \tfrac{1}{\epsi}V_v^{\epsi,\Lambda} +   H_1^\Lambda$ for the sum of the two types of perturbations we consider. The dynamical process of applying the perturbation is modelled by a time-dependent Hamiltonian 
$H^{\epsi,\Lambda }(\eta s) := H_0^\Lambda + f(\eta s)\,\epsi V^{ \Lambda} $, where   $f:\R\to \R$ is a 
 sufficiently regular switch function with $f(t)=0$ for $t\leq -1 $ and $f(t)=1$ for $t\geq 0$.
Thus, we have $H^{\epsi,\Lambda }(\eta s) = H_0$ for all $s\leq -\frac{1}{\eta}$ and $H^{\epsi,\Lambda }(\eta s) = H_0^\Lambda +  \epsi V^{ \Lambda} $ for all $s\geq 0$. The parameter $\eta>0$ controls the time scale $\frac{1}{\eta}$ on which the switching process occurs.
The state of the system at  time $s\geq 0$, when the perturbation is fully switched on, is obtained from solving  the time-dependent Schr\"odinger equation
$
\I\,    \frac{\D}{\D s}  \,  \tilde\rho^{\epsi , \Lambda,\eta,f}(s) = [ H^{\epsi,\Lambda}(\eta s)  ,  \tilde\rho^{\epsi, \Lambda,\eta,f}(s)] 
$
with initial condition $\tilde\rho^{\epsi, \Lambda,\eta,f}(s)= P_*^\Lambda$ for all $s\leq -\frac{1}{\eta}$.
It is convenient and common practice to rescale the time variable to $t=\eta s$, which yields the standard form of the Schr\"odinger equation with adiabatic parameter $\eta>0$, namely
\[
\I\, \eta   \tfrac{\D}{\D t}  \,   \rho^{\epsi , \Lambda,\eta,f}(t) = [ H^{\epsi,\Lambda}(  t)  ,   \rho^{\epsi, \Lambda,\eta,f}(t)] \,
\]
with $\rho^{\epsi , \Lambda,\eta,f}(t) = \tilde \rho^{\epsi , \Lambda,\eta,f}(t/\eta)$. 

The response of the system with respect to an observable $B^\Lambda$ at time $t\geq 0$ is now defined as 
\[
\sigma^{\epsi,\Lambda,\eta,f}(t) := \tfrac{1}{N(\Lambda)} \left( 
\left\langle B \right\rangle_{\rho^{\epsi,\Lambda,\eta,f}(t)} -
\left\langle B \right\rangle_{P_*^\Lambda} 
\right)\,,
\]
i.e.\ as the difference between the expectation of $B^\Lambda$ in the state after the perturbation $V^{ \Lambda}$ was turned on and its expectation in the initial ground state $P_*^\Lambda$ of the unperturbed Hamiltonian~$H_0^\Lambda$.
Here $N(\Lambda)$ is a normalisation depending on the localisation properties of the perturbation $V^{ \Lambda}$ and the observable  $B^\Lambda$. E.g., in the case of extensive quantities $B^\Lambda$ and perturbations $V^{ \Lambda}$ that act everywhere, $N(\Lambda)$ must be chosen proportional to the volume $|\Lambda|$.
Since in general many-body situations the quantity $\sigma^{\epsi,\Lambda,\eta,f}(t)$ is neither computable nor interesting, one considers the following asymptotic regimes,  that lead to explicit and practically useful formulas. First, since one is interested in macroscopic systems, to avoid finite size effects one takes the 
{\em thermodynamic limit} $\Lambda\to\Z^d$.
Second,  one expects that in the {\em adiabatic limit} $\eta\ll1$ of  slow switching the system settles in an (almost) stationary state that has no ``memory'' of the switching procedure, i.e.\ that  the response becomes independent of the precise form of the switching function~$f$ and also independent of the time $t\geq 0$. 
Finally, one is interested in {\em small perturbations} and thus in an expansion of the response in powers of~$\epsi $.

The standard linear response calculation   now proceeds by expanding 
$\sigma ^{\epsi,\Lambda,\eta,f}(0)$ first in powers of $\epsi$,
\begin{equation}\label{standard1}
\sigma ^{\epsi,\Lambda,\eta,f}(0) =:  \epsi \,\sigma^{\Lambda,\eta,f}_1 +  o_{\Lambda,\eta,f}(\epsi)\,,
\end{equation}
where $o_{\Lambda,\eta,f}(\epsi)$ denotes a remainder  term that is $o(\epsi)$   when $\Lambda$, $\eta$, and $f$ are kept fixed.
This expansion can be easily achieved by standard time-dependent perturbation theory and yields 
\begin{equation} \label{KuboL} 
 \sigma^{\Lambda,\eta,f}_1=
   -\tfrac{ \I}{N(\Lambda)}  \int_{-\infty}^0 f(\eta s)\, \left\langle \left[ B  (-s) ,  V \right]
  \right\rangle_{P_*^\Lambda} \,\D s\,.
\end{equation} 
Here we use the notation  $B^\Lambda(s) :=  \E^{ \I H_0^\Lambda s} B^\Lambda  \E^{-\I H_0^\Lambda s}$ for the Heisenberg time-evolution of an operator $B^\Lambda$.

Then one considers $\sigma^{\Lambda,\eta,f}_1$ in the adiabatic and thermodynamic limit and calls the resulting quantity
  \begin{equation}\label{standard2}
 \sigma_1 := \lim_{\eta\to 0} \lim_{\Lambda\to \Z^d} \sigma^{\Lambda,\eta,f}_1 =   -\I  \lim_{\eta\to 0}    \int_{-\infty}^0 f(\eta s)\, \left\langle \left[ B (-s) , V \right]
  \right\rangle_{P_*^\infty} \,\D s
  \end{equation}  
the {\em linear response coefficient}. Typically, if the limits exist, $\sigma_1$ is independent of $f$ and one chooses $f(\eta s) = \E^{\eta s}$ to simplify the explicit evaluation of \eqref{standard2}.
  
However, the procedure just described  is only justified when the remainder term in \eqref{standard1} is of lower order in $\epsi$ {\em uniformly} in  the volume $\Lambda$ and in the adiabatic parameter~$\eta$. 
This uniformity of time-dependent perturbation theory on long (adiabatic) time scales  is anything but obvious and can only be expected if the system indeed evolves into  an (almost) stationary state that looses all memory of the switching process. But the occurrence of such a  behaviour cannot hold unconditionally   and needs to be established under suitable additional assumptions.

Recently  Bachmann et al.\ \cite{BDF} were able to prove 
validity of linear response 
 for interacting spin systems and for  quasi-local perturbations 
 that {\em do not close the spectral gap} (i.e., they assume that $H^{\epsi,\Lambda }(t)$ has a spectral gap above its ground state $P_*^{\epsi,\Lambda}(t)$  uniformly in $\epsi\in [0,\epsi_0)$, $\Lambda$, and $t\in\R$).
 The key ingredient is their adiabatic theorem that shows that
 in the adiabatic limit $\eta\to 0$  for local observables $B$
 the expectation $\left\langle B \right\rangle_{\rho^{\epsi,\Lambda,\eta}(0)}  $ converges to $\left\langle B \right\rangle_{P_*^{\epsi,\Lambda}(0)}  $, i.e.\ that
 \[
\sigma^{\epsi,\Lambda }(0):=  \lim_{\eta \to 0 } \sigma^{\epsi,\Lambda,\eta,f}(0) =  \left\langle B \right\rangle_{P_*^{\epsi,\Lambda}(0)}- \left\langle B \right\rangle_{P_*^\Lambda}  
 \]
 uniformly in the volume $\Lambda$. Then, using the spectral flow, they show that the asymptotic expansion of $\sigma^{\epsi,\Lambda }(0)$   in powers of $\epsi$ starts with a linear term,
 \[
 \sigma^{\epsi,\Lambda }(0) = \epsi \,\sigma^\Lambda_1 + o (\epsi)\,,
 \]
 where, again, the error term is uniform in $\Lambda$.  
  As all limits are uniform in the volume $|\Lambda|$, one can take the thermodynamic limit in the end and, if it exists, it defines the physically meaningful linear response coefficient
  \[
  \tilde \sigma_1 := \lim_{\Lambda\to \Z^d}  \sigma^\Lambda_1\,.
  \]
  While Bachmann et al.\ \cite{BDF} do not discuss this in detail, it is likely that under rather mild assumptions one can show that $\tilde \sigma_1$ agrees with $  \sigma_1$ in \eqref{standard2}.
  
 We now show that, using the results presented in the previous section, a similar reasoning allows us to rigorously derive linear and also higher order response coefficients even for perturbations that close the gap. However, one important conceptual difference appears. Since the system evolves into an almost-invariant state with a life-time that depends on the strength $\epsi$ of the perturbation,  the   adiabatic switching must   occur on a  time-scale not longer than the life-time, cf.\  Remark~\ref{timescalerem}. 
 This puts a lower bound of the form $\eta\geq \epsi^m$ for some $m\geq 1$ on the adiabatic parameter $\eta$. On the other hand, as we will show, even a relatively fast switching with $\eta \leq \epsi^\frac{1}{m}$ for some $m\geq 1$ still leads to a response coefficient that has an asymptotic expansion in powers of $\epsi$ with coefficients independent of the switching function $f$.

\begin{theorem}[\bf Response theory to all orders]\label{response}
Under the same assumptions as in Proposition~\ref{prop:switch}
let again $\rho^{\epsi , \Lambda,\eta,f}(t)$ be the solution of the adiabatic time-dependent Schr\"odinger equation 
 \eqref{Schroedinger2}  with adiabatic parameter $\eta\in (0,1]$  and initial datum $\rho^{\epsi, \Lambda,\eta,f}(t) = P_*^\Lambda$ for all $t\leq -1$.

 For $L\in\loc$ with $\ell\cdot\ell_H=0$ and $B\in \mathcal{L}_{\mathcal{S},\infty,L_1}$   define   the normalised  response  as
 \[
 \sigma^{\epsi,\Lambda,\eta,f}(t) :=  \tfrac{\epsi^{|\ell_H|\gamma}}{M^{d-|\ell|-|\ell_H|}} \, \left(\left\langle B \right\rangle_{\rho^{\epsi,\Lambda,\eta,f}(t)} -
\left\langle B  \right\rangle_{P_*^\Lambda} 
\right)\,,
 \]
 and for $j\in\N$ the $j$th order response coefficient as
 \[
 \sigma^\Lambda_j := \tfrac{\epsi^{|\ell_H|\gamma}}{M^{d-|\ell|-|\ell_H|}} 
\left\langle \mathcal{K}^{\epsi }_j [B ] \right\rangle_{P_*^\Lambda} 
 \,,
 \]
 where the $\mathcal{K}^{\epsi,\Lambda}_j $'s were defined and explicit expressions for $\left\langle \mathcal{K}^{\epsi }_1[B ]  \right\rangle_{P_*^\Lambda} $ and $\left\langle \mathcal{K}^{\epsi }_2 [B ]  \right\rangle_{P_*^\Lambda} $ were given in Proposition~\ref{proposition:expand1}.
 
 Then for any $n \in\N$  and $m\geq 1$ there exists a constant $C\in\R$ independent of   $\epsi$, such that for  $t\geq 0$ and $r= \max\{2d,n+1\}$
 \begin{equation}\label{eq:response}
\sup_{\Lambda(M):M\geq M_0}\; \sup_{\eta\,\in\left[\epsi^m,\,\epsi^\frac{1}{m}\right]} \left|  \sigma^{\epsi,\Lambda,\eta,f}(t) 
 -   \sum_{j=1}^n \epsi^j  \sigma^\Lambda_j
 \right| \leq \epsi^{n+1}\,C\,(1+t^{d+1})\,\|\Phi_B\|_{\zeta_r,r,L_1}  \,.
 \end{equation}
\end{theorem}  

\begin{remark}\rm
\balp
\item The response coefficients $\sigma_j^\Lambda$ are independent of the switch  function $f$, the adiabatic parameter $\eta$, time $t\geq 0$, and also of $n$ and  $m$. For slowly varying potentials of the form $v^{\epsi,\Lambda}= \epsi v^\Lambda$ or perturbations of the form $\epsi H_1$, they are also independent of $\epsi$. 
Thus the   response   has an asymptotic expansion in the strength $\epsi$ of the perturbation that is uniform in the volume and with coefficients that are constant in time and have no memory of the switching.
\item 
Looking at the first two coefficients $\sigma_1^\Lambda$ and $\sigma_2^\Lambda$ given by \eqref{K1} and \eqref{K2}, one notes that they agree with what one would get from just applying regular perturbation theory for isolated eigenvalues, e.g.\ \cite{K}. This is true also for all higher order terms, as follows from Remark~\ref{remark:NEASS} (b).
\item Another formula for $\sigma_1^\Lambda= \tfrac{\epsi^{|\ell_H|\gamma}}{M^{d-|\ell|-|\ell_H|}}  \left\langle \left[  \left[V , R _0\right],B 
\right]\right\rangle_{P_*^\Lambda}
$ that can be checked by direct computation is
\[
\sigma_1^\Lambda =   - \I\,\tfrac{\epsi^{|\ell_H|\gamma}}{M^{d-|\ell|-|\ell_H|}} \lim_{\eta\searrow 0}  \int_{-\infty}^0 \E^{\eta s} \, \left\langle \left[ B (-s) , V \right]
  \right\rangle_{P_*^\Lambda} \,\D s\,,
\]
i.e.\ one chooses $f=\exp$ in \eqref{KuboL}. 
 \ealp\label{linrem}
\end{remark}  
  
   \begin{proof}[Proof of Theorem~\ref{response}.]
  According to Proposition~\ref{prop:switch},  for $\tilde n := m(n+ d+1)$ there exists $C_1\in\R$ such that    for $t\geq 0$ and $\epsi^\frac{1}{m}\geq \eta\geq \epsi^m$  it holds  that
  \begin{align*}
 \sup_{\Lambda}  \frac{\epsi^{|\ell_H|\gamma}}{M^{d -|\ell|-|\ell_H|}} \left| \left\langle B \right\rangle_{\rho^{\epsi,\Lambda,\eta,f}(t)} 
 - \right.&\left.
 \left\langle B \right\rangle_{ \Pi_{\tilde n}^{\epsi, \Lambda} } 
\right| \;\leq\\ &\leq \; C_1 \,  \frac{\epsi^{\tilde n+1} + \eta^{\tilde n+1}}{\eta^{d+1}}   \,(1+t^{d+1})\,  \|\Phi_B\|_{\zeta_{2d},2d,L_1} \\& \leq\;
   2 C_1 \,  \epsi^{n+1}   \, (1+t^{d+1})\, \|\Phi_B\|_{\zeta_{2d},2d,L_1} \,.
   \end{align*}
   According to Proposition~\ref{proposition:expand1},
   there exists $C_2\in\R$ such that 
   \[
   \sup_\Lambda  \frac{\epsi^{|\ell_H|\gamma}}{M^{d-|\ell|-|\ell_H|}} \left|  \, \left\langle B \right\rangle_{ \Pi_{\tilde n}^{\epsi, \Lambda} }  -\sum_{j=0}^n \epsi^j\,
\left\langle \mathcal{K}^{\epsi }_j (B ) \right\rangle_{P_*^\Lambda}\right| \;=\; C_2\,\epsi^{n+1}\, \|\Phi_B\|_{\zeta_{n+1},n+1,L_1}  \,,
   \]
   which proves \eqref{eq:response}.     
   \end{proof}
  
While we believe that our results are of general conceptual interest, let us  end this section with a concrete example, where linear response for gap closing potentials   can play an important role.
The understanding of the fractional quantum Hall effect \cite{L} rests on the idea that the many-body interaction itself can produce  spectral gaps at fractional fillings of bands. These gaps are much  smaller than gaps between Landau levels and not much is known about their stability under external perturbations.
While the stability of the spectral gap for small perturbations of free fermions is now well understood \cite{H,DS} (and the results of \cite{BDF} can be applied at least for sufficiently small perturbations), for general many-body ground states, like those in the fractional Hall effect, such a result seems out of reach at the moment. Moreover,    for the very small gaps at fractional filling it could  hold only for very small perturbations anyway.
However, as our results show,  the validity of linear response for small and slowly varying perturbations does not depend on the global stability of spectral gaps.

\section{The space-time adiabatic theorem} \label{sec:spacetime}

In this section we consider time-dependent Hamiltonians of the form 
\[
H(t) = H_0(t) + V_v(t) + \epsi H_1(t)\,,
\]
 where only $H_0(t)$ is assumed to have  a gapped ground state, and prove a novel adiabatic theorem of the following type. While the standard adiabatic theorem asserts that solutions of the time-dependent Schr\"odinger equation remain near gapped spectral subspaces of the instantaneous Hamiltonian in the adiabatic limit, our theorem states that solutions   remain near non-equilibrium almost-stationary states of the instantaneous Hamiltonian.
The results of this section contain  the statements of Section~\ref{sec:NEASS} as special cases.
 
\begin{definition}[\bf Spaces of time-dependent operators]\label{def:sb}
Let $\timedom\subseteq \R$ be an interval.
We say that a map $A:\timedom\to \mathcal{L}_{\zeta,n,L_\alpha}$ is   smooth and bounded whenever it is given by interactions $\Phi_A(t)$ such that
the maps
\[
\timedom \to \mathcal{A}_X^\mathfrak{N} \,,\quad t\mapsto \Phi^\Lambda_A(t,X) \,,
\]
are infinitely differentiable for all $\Lambda$ and $X\subset \Lambda$ and
\[
\sup_{t\in\timedom} \|\Phi_A^{(k)}(t)\|_{\zeta,n,L_\alpha} <\infty \qquad\mbox{for all }\quad k\in\N_0\,.
\]
Here $\Phi_A^{(k)}(t) = \{ (\tfrac{\D^k}{\D t^k}\Phi^{\Lambda }_A) (t)\}_\Lambda$
denotes the interaction defined by the term-wise derivatives. 
The corresponding spaces of smooth and bounded time-dependent interactions and operator-families are denoted by $ \mathcal{B}_{I,\zeta,n,L_\alpha}$ and $ \mathcal{L}_{I,\zeta,n,L_\alpha}$.

We say that $A:\timedom\to \mathcal{L}_{\mathcal{S},\infty,L_\alpha}$ is   smooth and bounded if for any $n\in\N_0$ there is a $\zeta_n\in\mathcal{S}$ such that $A:\timedom\to \mathcal{L}_{\zeta_n,n,L_\alpha}$ is   smooth and bounded and write $ \mathcal{L}_{I,\mathcal{S},\infty,L_\alpha}$ for the corresponding space.
\end{definition}

\noindent {\bf (A1)$_{\timedom,L^H}$  Assumptions on $H_0$.}\\[1mm] {\em
Let $\timedom\subseteq\R$ be an interval and
let $H_0:\timedom\to \mathcal{L}_{a,\infty}$ be smooth and bounded, i.e.\ $H_0 \in  \mathcal{L}_{I,a,\infty}$, with values in the self-adjoint operators.
If $L^H\in\loc$ is non-trivial, i.e.\ $|\ell_H|>0$, then we assume in addition  that $\frac{\D}{\D t} H_0 \in  \mathcal{L}_{I,a,\infty,L^H_1}$, i.e.\ that its time-derivative is $L^H_1$-localised.

Moreover, we assume that there exists $M_0\in\N$ such that for all $t\in \timedom$,  $M\geq M_0$ and corresponding $\Lambda=\Lambda(M)$  the operator $H_0^\Lambda(t)$   has a gapped part $\sigma^\Lambda_*(t)\subset\sigma(H^\Lambda_0(t))$ of its spectrum with the following properties: There exist continuous functions $f^\Lambda_\pm:\timedom\to \R$ and   constants $g>\tilde g >0$ such that $f^\Lambda_\pm(t)\in \rho(H_0^\Lambda(t))$,
$f^\Lambda_+(t)- f^\Lambda_-(t) \leq \tilde g$,
\[
[f^\Lambda_-(t), f^\Lambda_+(t)] \cap \sigma(H^\Lambda_0(t)) = \sigma^\Lambda_*(t)\,,\quad\mbox{and}\quad
{\rm dist}\left( \sigma_*^\Lambda(t) , \sigma(H_0^\Lambda(t))\setminus \sigma_*^\Lambda(t)\right) \geq g
\]
for all $t\in\timedom$. We denote by $P_*^\Lambda(t)$ the spectral projection of $H_0^\Lambda(t)$ corresponding  
to the spectrum $\sigma_*^\Lambda(t)$. }
\bigskip

Note that now also the time dependence of $H_0(t)$ is part of the perturbation and thus, possibly, a localisation of this part of the driving can be relevant.

In contrast to the previous section, we now require merely that $H_0(t)$ has a ``narrow'' gapped part $\sigma_*(t)$ of the spectrum which need not consist of a single eigenvalue. Typically we have in mind almost-degenerate ground states consisting of a narrow group of finitely many eigenvalues separated by a gap from the rest of the spectrum. The condition that the width $\tilde g$ of the spectral patch $\sigma_*(t)$ is smaller than the gap $g$ enters in the proof through the requirement that all operators appearing in the statements are quasi-local. More technically speaking, we need the range of the quasi-local inverse $\mathcal{I}_H$ of the Liouvillian have a vanishing $P_*(\cdots)P_*$-block, cf.\ Appendix~\ref{app:qli}.

\bigskip

\noindent {\bf (A2)$_{\timedom,L^H,\gamma}$  Assumptions on the perturbations.}\\[1mm]{\em
Let $H_1\in \mathcal{L}_{\timedom,\mathcal{S},\infty,L^H_1}$ be self-adjoint and   
let $v:\timedom\to \SVP$ be a time-dependent slowly varying potential such that the functions $t\mapsto v^{\epsi,\Lambda}(t,x)$ are infinitely differentiable for all $\epsi\in(0,1]$, $\Lambda$, and $x\in \Lambda$. Moreover, assume that all derivatives 
$\frac{\partial^k}{\partial t^k} v^{\epsi,\Lambda}(t,x)$ define again slowly varying potentials. Denote by $V_v$ the corresponding operator-family.

 If $L^H$ is non-trivial, then assume also  that $\frac{\D^k}{\D t^k}[H_0(t), \frac{1}{\epsi}V_v(t)]\in \mathcal{L}_{\mathcal{S},\infty,  L_{\epsi^\gamma}^H}$  uniformly in $t\in\timedom$ for all $k\in\N_0$.
}
\bigskip

Recall that  the time-independent  NEASS  of the previous section was given by conjugation of the spectral projection $P_*$ of $H_0$ by the time-independent unitary $\E^{\I  \epsi  S } $. In the time-dependent setting the spectral projection $P_*(t)$, the analog $\Pi_{*,n}^\delta(t) = \E^{\I  \epsi  S^{ \delta}_n(t)}   \,  P_* (t) \,\E^{-\I  \epsi  S^{ \delta}_n(t)}$  of the NEASS, and the unitary $ \E^{\I  \epsi  S^{ \delta}_n(t)}  $ connecting them, all depend on time. 
A special case of our adiabatic theorem  is as follows.  Assume that the system starts at time $t_0$ in the state $\rho(t_0) = \Pi_{*,n}^\delta(t_0)$, i.e.\ the initial state is the projection onto the full superadiabatic subspace. Then the 
solution of the time-dependent Schr\"odinger equation 
\begin{equation}\label{Schoedinger1}
\I \epsi\delta \tfrac{\D}{\D t} \rho(t) = [H(t) , \rho(t)]  
\end{equation}
remains close to   $\Pi_{*,n}^\delta(t)$ up to small errors in $\epsi$ and for long times.
Note that   we replaced the adiabatic parameter $\eta$ of the previous sections by $\eta= \epsi\delta$ with  $\delta\in[0,\frac{1}{\epsi}]$ in order to facilitate the notation in asymptotic expansions in powers of $\epsi$. In particular, $\Pi_{*,n}^0(  t)$ is the static NEASS of $H(  t)$ discussed in the previous section.
We will show that at times $\tilde t\in \timedom$ where the first $n$ time-derivatives of $H$ vanish, we have indeed that $\Pi_{*,n}^\delta(\tilde t)$  equals the static NEASS  $\Pi_{*,n}^0(\tilde t)$. If, in addition, also the perturbation vanishes at time $\tilde t$, i.e.\ $H(\tilde t) = H_0 (\tilde t)$, then $\Pi_{*,n}^0(\tilde t)= P_*(\tilde t)$.

However, the general statement concerns an adiabatic approximation   to solutions of \eqref{Schoedinger1}
with initial condition given by an arbitrary state   
\[
\Pi_{n}^\delta(t_0) := \E^{\I  \epsi  S^{ \delta}_n(t_0)}   \,  \rho_0   \,\E^{-\I  \epsi  S^{ \delta}_n(t_0)}
\]
in the range of $\Pi_{*,n}^\delta(t_0)$, i.e.\ $\rho_0$ is a state satisfying  $\rho_0 = P_*(t_0)\rho_0P_*(t_0)$. Then the adiabatic approximation must also track the evolution within the spectral subspace and the statement becomes  that 
\[
\Pi_{n}^\delta(t ) := \E^{\I  \epsi  S^{ \delta}_n(t )}   \,  P_n^\delta (t ) \,\E^{-\I  \epsi  S^{ \delta}_n(t )}
\]
is close to the solution of \eqref{Schoedinger1} with initial condition $\rho(t_0) = \Pi_{ n}^\delta(t_0)$, if the state $P_n^\delta (t )$ satisfies the effective equation 
\begin{equation}\label{effectiveTrans0}
 \I \epsi\delta \tfrac{\D}{\D t} P_n^{\delta}(t)  =   \big[ \epsi \delta K (t) +\sum_{\mu=0}^n \epsi^\mu h_\mu^{ \delta}(t) ,P_n^{ \delta}(t)\big]  
 \quad\mbox{ with } \quad P_n^\delta(t_0) = \rho_0\,.
\end{equation}
 Here $K(t):=   \mathcal{I}_{H_0(t),g,\tilde g}(\dot H(t) ) $ (cf.\ Appendix~\ref{app:qli}) is a self-adjoint quasi-local generator   of the parallel transport within the vector-bundle $\Xi_{*,\timedom}$ over $\timedom$ defined by $t\mapsto P_*(t)$, i.e.\ its off-diagonal terms are
 \[
 P_*(t) K(t) P_*(t)^\perp = \I P_*(t) \dot P_*(t)  \quad\mbox{ and } \quad P_*(t)^\perp K(t) P_*(t)  = -  \I \dot P_*(t)P_*(t) \,,
 \]
 and $P_*(t) K(t)P_*(t) = 0$.
 The coefficients $h^\delta_\mu(t)$ of the effective Hamiltonian generating the dynamics within $\Xi_{*,\timedom}$ are quasi-local operators that commute with $P_*(t)$. Thus, for $\rho_0 = P_*(t_0)$  the unique solution of \eqref{effectiveTrans0} is $P_*(t)$ and the effective Hamiltonian $\sum_{\mu=0}^n \epsi^\mu h_\mu^{ \delta}(t)$ is irrelevant.
\medskip
 
 Our main results are now stated in three steps. Proposition~\ref{proposition:Adi} contains the detailed properties of the super-adiabatic state that solves the time-dependent Schr\"odinger equation up to a small and quasi-local residual term.
The adiabatic theorem,  Theorem~\ref{theorem:Adi}, asserts that the super-adiabatic state and the actual solution of the time-dependent Schr\"odinger equation are close in the sense that they   agree on expectations for quasi-local observables up to terms asymptotically smaller than any power of $\epsi $. Finally, in Proposition~\ref{proposition:expand} we determine the asymptotic expansion of the NEASS.

\begin{proposition}[\bf The space-time adiabatic expansion]\label{proposition:Adi}
For some $L^H\in\loc$ and $\gamma\in\{0,1\}$ let the Hamiltonian $H_0$ satisfy Assumption (A1)$_{\timedom,L^H}$  and let $H= H_0 + V_v + \epsi H_1$ with $V_v$ and $H_1$ satisfying (A2)$_{\timedom,L^H,\gamma}$.  There exist two sequences   $(h^{ \delta}_\mu)_{\mu\in\N_0}$ and $(A^{ \delta}_\mu)_{\mu\in\N}$ of quasi-local $L^H_{\epsi^\gamma}$-localised operator families, 
such that for any $n\in\N_0$, $t_0\in I$,  and any family of initial states $\rho_0^\Lambda$ with $\rho_0^\Lambda = P_*^\Lambda(t_0) \rho_0^\Lambda P_*^\Lambda(t_0)$,   the state
\[
  \Pi_n^{\epsi,\Lambda,\delta}(t) := \E^{\I  \epsi  S^{\epsi,\Lambda,\delta}_n(t)}   \,  P_n^{\epsi,\Lambda,\delta}(t) \,\E^{-\I  \epsi  S^{\epsi,\Lambda,\delta}_n(t)}\,,
  \]
   solves 
\[
\I\epsi \delta\tfrac{\D}{\D t}    \Pi_n^{\epsi,\Lambda,\delta}(t) = \left[ H^{\epsi,\Lambda}(t),  \Pi_n^{\epsi,\Lambda,\delta}(t)\right] + \epsi^{n+1}(1+\delta^{n+1})\left[R_{n}^{\epsi,\Lambda,\delta}(t),  \Pi_n^{\epsi,\Lambda,\delta}(t)\right]\,, 
\]
where
 \[
   S^{\epsi,\Lambda,\delta}_n(t) := \sum_{\mu = 1}^{n} \epsi^{\mu-1}  A^{\epsi,\Lambda,\delta}_\mu(t)  \,,
 \]
  $P_n^{\epsi,\Lambda,\delta}(t)$ is the solution of the effective equation
  \begin{equation}\label{effectiveTrans}
   \I \epsi\delta \tfrac{\D}{\D t} P_n^{\epsi,\Lambda,\delta}(t)  =   \Big[ \epsi \delta K^{ \Lambda}(t) +\sum_{\mu=0}^n \epsi^\mu h_\mu^{\epsi,\Lambda,\delta}(t) ,P_n^{\epsi,\Lambda,\delta}(t)\Big]  \quad \mbox{with  $P_n^{\epsi,\Lambda,\delta}(t_0)=  \rho_0^\Lambda$}  \,,
   \end{equation}
and $R_{n}^{ \delta}\in  \mathcal{L}_{I,\mathcal{S},\infty,L^H_{\epsi^\gamma}}$ uniformly for $ \delta\in[0,\frac{1}{\epsi}]$.
\medskip

\noindent Further properties of $(h^{ \delta}_\mu)_{\mu\in\N_0}$ and $(A^{ \delta}_\mu)_{\mu\in\N}$ are:
\begin{enumerate}\rom
\item  $[h^{\epsi,\Lambda,\delta}_\mu(t), P_*^\Lambda(t)] = 0$ for all $t\in\timedom$, $\mu\in\N_0$, and $\delta\in[0,\frac{1}{\epsi}]$.
\item 
All $A^{ \delta}_\mu$ are polynomials of   degree $\mu$ in $\delta$ with coefficients in $  \mathcal{L}_{I,\mathcal{S},\infty,L^H_{\epsi^\gamma}}$
and the constant coefficients $A^{ 0}_\mu(t)$
are those of the  almost-stationary state $ \Pi_n^{\epsi,\Lambda,0}(t)$ for the Hamiltonian $H^{\epsi,\Lambda}(t)$.
\item  
If for some $\tilde t \in\timedom$ it holds that $\frac{\D^k}{\D t^k} H (\tilde t )=0$ for $k=1,\ldots, n$, then
 $A^{ \delta}_\mu(\tilde t ) = A^{ 0}_\mu(\tilde t )$ and thus $ \Pi_n^{\epsi,\Lambda,\delta}(\tilde t )= \Pi_n^{\epsi,\Lambda,0}(\tilde t )$.
\end{enumerate}

\end{proposition}
\medskip

The proof of Proposition~\ref{proposition:Adi} is given in Section~\ref{sec:stae}.
Note that the generator $S_n(t)$ of the super-adiabatic transformation at time $t$ turns out to depend only on the Hamiltonian $H$ and its time-derivatives at   time $t$. As a consequence, if $\rho_0 = P_*(t_0)$, i.e.\ also $P_n^\delta(t) = P_*(t)$ depends only on $H_0$ at time $t$,
then the super-adiabatic projection $\Pi_n^\delta(t)$ itself depends only on the Hamiltonian $H$ and its time-derivatives at   time $t$.
In this case the NEASS   has no ``memory'' of the switching process.

We now state the general space-time adiabatic theorem that is proved in Section~\ref{sec:proofAdiabatic}.
 
 \begin{theorem}[\bf The space-time adiabatic theorem]\label{theorem:Adi}
For some $L^H\in\loc$ and $\gamma\in\{0,1\}$ let the Hamiltonian $H_0$ satisfy Assumption (A1)$_{\timedom,L^H}$  and let $H = H_0 + V_v + \epsi H_1$ with $V_v$ and $H_1$ satisfying (A2)$_{\timedom,L^H,\gamma}$.
  Let $\Pi_n^{\epsi,\Lambda,\delta}(t)$ be the super-adiabatic state of order $n>d$ constructed in Proposition~\ref{proposition:Adi} and $\rho^{\epsi,\Lambda,\delta}(t)$ the solution of the Schr\"odinger equation
  \[
\I\epsi \delta\tfrac{\D}{\D t}    \rho^{\epsi,\Lambda,\delta}(t) = [ H^{\epsi,\Lambda}(t) ,  \rho^{\epsi,\Lambda,\delta}(t)] \qquad 
  \mbox{with}\qquad  \rho^{\epsi,\Lambda,\delta}(t_0) = \Pi_n^{\epsi,\Lambda,\delta}(t_0)\,.
  \]
  Then for any $\zeta\in \mathcal{S}$ and $L\in\loc$ with $\ell\cdot\ell_H=0$ there exists a constant $C$ such that for any $B\in \mathcal{L}_{\zeta,2d,L_1}$, $\delta\in (0,\frac{1}{\epsi}]$,   and all   $t\in\timedom$
  \begin{align}\label{AdiState}\nonumber
\sup_{\Lambda(M):M\geq M_0} \;  \tfrac{\epsi^{|\ell_H|\gamma}}{M^{d -|\ell|-|\ell_H|}} &\left|  \tr   \left(\rho^{\epsi,\Lambda,\delta}(t)  B^\Lambda \right)- \tr  \left(  \Pi^{\epsi,\Lambda,\delta}_n(t)  B^\Lambda\right)\right| \\ & \leq \; \kappa  \,C \, \frac{\epsi^{n+1} + (\epsi\delta)^{n+1}}{(\epsi\delta)^{d+1}}   \, |t-t_0|(1+|t-t_0|^d) \,\|\Phi_B\|_{\zeta,2d,L_1}\,,
\end{align}
where $\dege := \sup_{\Lambda} \tr(\rho_0^\Lambda)$.
 \end{theorem}
 \medskip
   
 For simplicity, we state and prove the explicit expansion of the super-adiabatic state only for the
case that $\rho_0 = P_*(t_0)$ and thus   $P_n^\delta(t) = P_*(t)$. The general case could be handled as in Theorem~3.3 in \cite{MT}. The proof of the following statement is given in Section~\ref{proposition:expand:proof}.

     \begin{proposition}[\bf Asymptotic expansion of the time-dependent NEASS]\label{proposition:expand}
Under the assumptions of Theorem~\ref{theorem:Adi}
and if $\rho_0 = P_*(t_0)$, there exist time-dependent  linear maps $\mathcal{K}^{\epsi,\Lambda,\delta}_j: \mathcal{A}_\Lambda^\mathfrak{N} \to  \mathcal{A}_\Lambda^\mathfrak{N}$, $j\in\N$, such that for any $k\in\N_0$, $\zeta\in \mathcal{S}$, and $L\in\loc$ with $\ell\cdot\ell_H=0$ there is a constant $C$ such that
  for any $B\in \mathcal{L}_{\zeta,k+1,L_1}$ it holds that 
\begin{eqnarray*}\lefteqn{\hspace{-5mm}\nonumber
\sup_{t\in I} \sup_{\Lambda(M):M\geq M_0}\,  \tfrac{\epsi^{|\ell_H|\gamma}}{M^{d-|\ell|-|\ell_H|}} \left|  \,\tr\left( \Pi^{\epsi,\Lambda,\delta}_k(t) B^\Lambda\right) -   \sum_{j=0}^k \epsi^j\,\tr\left( P_*^\Lambda(t) \,\mathcal{K}^{\epsi,\Lambda,\delta}_j(t) [B^\Lambda] \right)\right| }\\ &&\hspace{55mm} \leq\;C\,\epsi^{k+1}(1+\delta^{k+1})\, \|\Phi_B\|_{\zeta,k+1,L_1} \,,
\end{eqnarray*}
with
\[
\sup_{t\in I}\sup_{\Lambda(M):M\geq M_0}\; \sup_{\epsi\in(0,1]}  \tfrac{\epsi^{|\ell_H|\gamma}}{M^{d-|\ell|-|\ell_H|}} \,\tr\left( P_*^\Lambda(t) \,\mathcal{K}^{\epsi,\Lambda,\delta}_j(t) [B^\Lambda] \right)<C(1+\delta^j)
\]
for all $j\leq k$.  

Each map $\mathcal{K}^{\epsi,\Lambda,\delta}_j$ is given by a finite sum of nested commutators with the operators $A^\delta_1,\ldots, A^\delta_j $  constructed in Theorem~\ref{SpaceAdiabaticThm}.
Explicitly, the first terms are 
 \[
  \mathcal{K}_0^{\epsi,\Lambda,\delta} (t) = {\rm Id} \,,\quad  \mathcal{K}_1^{\epsi,\Lambda,\delta}(t) [ \cdot] = - \I\, [A_1^{\epsi,\Lambda,\delta}(t) ,\,\cdot\,] \,, 
\]
\[
\mathcal{K}_2^{\epsi,\Lambda,\delta} (t) [\cdot] = -\I\, [A_2^{\epsi,\Lambda,\delta} (t),\,\cdot\,]
- \tfrac{1}{2}[A_1^{\epsi,\Lambda,\delta}(t) ,[A_1^{\epsi,\Lambda,\delta} (t),\,\cdot\, ]]\,.
\]
If $\sigma_*(t) = \{E_*(t)\}$ is a single eigenvalue, then
 \begin{eqnarray*}
\tr\left( P_*^\Lambda(t) \,\mathcal{K}^{\epsi,\Lambda,\delta}_1(t)\left[ B^\Lambda\right] \right) &=&   \I\delta\,\tr\left(  P_*^\Lambda(t) \left[\left( \dot P_*^\Lambda (t) R^\Lambda_0(t) + R^\Lambda_0(t)\dot P_*^\Lambda (t) \right),B^\Lambda\right]\right)\\&& +\;   
\tr\left( P_*^\Lambda(t) \,\left[  \left[R^\Lambda_0(t),\tfrac{1}{\epsi}V^{\epsi,\Lambda}_v(t)+H_1^{\epsi,\Lambda}(t)\right], B^\Lambda 
\right]\right)
 \end{eqnarray*}
where 
 $R^\Lambda_0(t) := (H_0^\Lambda(t) - E_*^\Lambda(t))^{-1} ( 1- P_*^\Lambda(t))$ denotes the reduced resolvent of $H_0(t)$. In the time-independent case, the second order term is given by \eqref{K2}.
\end{proposition}

\section{Proofs of the main results} 
 
\subsection{The space-time adiabatic expansion} \label{sec:stae}

 \begin{proof}[Proof of Proposition~\ref{proposition:Adi}.]
To simplify the notation and to improve readability, we   drop  all super- and subscripts that are not necessary to distinguish different objects, as well as the dependence on time, $\epsi$, and $\Lambda$.

Taking a time derivative of  $\Pi  = U   P_n U^* $,  where $U:= \E^{\I \epsi S}$, and using \eqref{effectiveTrans} yields
\begin{align*}
\I\delta \epsi  \frac{\D}{\D t}  \Pi  &= \I  \delta \epsi \dot U  P_n U^* + \I  \delta\epsi U  P_n \dot U ^* +   U  [ \epsi \delta K +h , P_n ] U ^*\nonumber\\
&= [ H,  \Pi ] + U  \left[ \I\epsi\delta U ^*\dot U  +   ( \epsi \delta K+h)     - U ^*HU ,  P_n \right] U ^*\,,
\end{align*}
where for the second equality we employed the identities  $U \dot U ^* = -\dot U  U ^*$   and $  [ H,  \Pi ] = U   [U ^*HU , P_n ] U ^* $.
We thus need to choose the coefficients $A_\mu$ entering in the definition of $U $ in such a way that the remainder term satisfies
\begin{align}\label{Rdef}\nonumber
U  \left[ \I\epsi  \delta U ^* \dot U +   \epsi \delta K+h    - U ^*H_0U\right.&\left.   - \,U ^*(V_v +\epsi H_1)U,P_n  \right] U ^* =\\
&=\;
\epsi^{n+1}(1+\delta^{n+1}) U [ U^* R_n U ,P_n] U^*  \,.
\end{align}
Expanding $U ^*H_0U $ yields
\begin{align*}\nonumber
U ^*H_0 U &=  \E^{-\I  \epsi  S  } H_0\, \E^{\I  \epsi  S } =  \sum_{k=0}^n \frac{ \epsi ^k}{k!} \,\ad_{S }^k(H_0)
+  \frac{ \epsi ^{n+1}}{(n+1)!}  \E^{-\I \tilde\gamma S  }  \, \ad_{S  }^{n+1} (H_0) \,\E^{\I  \tilde \gamma S  }\\
&=: \sum_{\mu=0}^n  \epsi^\mu H_{0,\mu} +\epsi^{n+1}H_{0,n+1} \,,
\end{align*}
where  $\tilde \gamma\in [0,\epsi]$, each $H_{0,\mu}$, $\mu=1,\ldots,n$, is defined as  the sum of those terms in the series that carry a factor $\epsi^\mu$, and  $\epsi^{n+1}H_{0,n+1}$ is the sum of all remaining terms. Above we use the notation $\ad_A(B):= - \I [A,B]$ and 
$\ad_{S }^k(H_0) $ for the nested commutator $[-\I S  ,[\cdots, [-\I S  ,[-\I S  ,H_0 ]]\cdots]]$, where $-\I S  $ appears $k$ times. 
Clearly
\[
H_{0,\mu} =-  \ad_{H_0} (A_\mu) + L_\mu\,,
\]
where $L_\mu$ contains a finite number of iterated commutators of the operators $A_\nu$, $\nu<\mu$, with $H_0$.
Explicitly, the first orders are 
\[
H_{0,0} = H_0\,, \qquad H_{0,1} = -\ad_{H_0} (A_1)\,,\qquad H_{0,2} =- \ad_{H_0} (A_2) - \tfrac{1}{2} [A_1,[A_1,H_0]]\,.
\]
We obtain a similar expansion for 
\[
U ^*V_v  U   =: \sum_{\mu=0}^{n-1}  \epsi^\mu V_\mu +\epsi^{n+1}V_n \,,
\qquad\mbox{and}\qquad
\epsi \,U ^*H_1  U   =: \sum_{\mu=0}^{n-1}  \epsi^{\mu+1} H_{1,\mu} +\epsi^{n+1}H_{1,n} \,,
\]
with 
\[
V_0 = V_v\,, \qquad V_1 =  -\ad_{V_v} (A_1)\,,\qquad V_2 =  -\ad_{V_v} (A_2) - \tfrac{1}{2} [A_1,[A_1,V_v ]]\,.
\]
As a crucial observation is now, that, according to Lemma~\ref{lemma:Vcomm},    if the operators $A_\nu$ for $\nu<\mu$ are  polynomials in $\delta$ of degree $\nu$ with coefficients in $\mathcal{L}_{I,\mathcal{S},\infty,L^H_{\epsi^\gamma} }$, then the operator $V_\mu$ for $\mu\geq 1$  is   of the form
\[
V_\mu = \epsi \tilde V_\mu \qquad   \mbox{ with $\tilde V_\mu$ a polynomial of degree $\mu$ with coefficients in  $\mathcal{L}_{\mathcal{S},\infty,L^H_{\epsi^\gamma} }$} \,.
\]
We obtain the expansion of  $U ^*\dot U $   by expanding the integrand of Duhamel's formula 
\[
\I \epsi U^*\dot U = -\epsi^2 \int_0^1 \E^{-\I \lambda \epsi S } \,\dot S   \,\E^{\I \lambda \epsi S }\,\D \lambda 
\]
as a Taylor polynomial  of nested commutators  and then integrating term by term, 
\begin{align*}
\I \epsi U ^*\dot U &= -\epsi^2    \sum_{k=0}^{n-2} \frac{ \epsi   ^k}{(k+1)!} \ad_{S }^k(\dot S)
- \frac{  \epsi ^{n+1}}{(n-1)!} \int_0^1 \E^{-\I \lambda \tilde\gamma S  }   \ad_{S  }^{n-1} (\dot S ) \E^{\I \lambda \tilde \gamma S  }\,\D \lambda
\\
& = \sum_{\mu=1}^{n } \epsi^\mu Q_\mu + \epsi^{n+1} Q_{n+1} \,.
\end{align*}
Here, again, $Q_\mu$ collects all terms in the sum proportional to $\epsi^\mu$. Note that $Q_\mu$ is a finite sum of iterated commutators of the operators $A_\nu$ and $\dot A_\nu$ for $\nu<\mu$, the    first terms being
\[
Q_1 = 0\,,\qquad  Q_2 = -\dot A_1 \,,\qquad Q_3 = - \dot A_2 + \tfrac{\I}{2} [ A_1,\dot A_1]\,.
\]
Inserting the expansions into \eqref{Rdef} leaves us with
\begin{align*} 
 \I\epsi \delta U ^* \dot U +  \epsi \delta K+h   & - U ^*H_0U   -   U ^*(V_v +\epsi H_1)U  \;=\\[2mm]
 &=\;\epsi (  \delta   K  + \tfrac{1}{\epsi} (h_0 - H_{0} -V_v) + h_1 -H_{0,1}  - H_{1,0}) \\
 &\; \;+\;\sum_{\mu=2}^n \epsi^\mu (\delta Q_\mu  +   h_{\mu}  -H_{0,\mu} - \tilde V_{\mu-1} - H_{1,\mu-1})  \\
 & \;\;+\; \epsi^{n+1} (Q_{n+1}  - H_{0,n+1}  - \tilde V_n - H_{1,n}  )\,.
\end{align*}
It remains to determine $A_1,\ldots, A_n$ inductively such that
\begin{equation}\label{Solve1}
[ \ad_{H_0} (A_1)+ \delta  K  + \tfrac{1}{\epsi}( h_0 -H_0 -V_v)+ h_1  -    H_{1,0}, P_n] \stackrel{!}{=} 0
\end{equation}
and
\begin{equation}\label{Solve2}
 [ \ad_{H_0} (A_\mu)+ \delta Q_\mu  +   h_\mu -L_\mu - \tilde V_{\mu-1} - H_{1,\mu-1}, P_n]
 \stackrel{!}{=} 0
\end{equation}
for all $\mu = 1,\ldots, n$. Equations~\eqref{Solve1} and \eqref{Solve2} can be solved for $A_\mu$, since $L_\mu$,  $Q_\mu$, $\tilde V_{\mu-1}$,   and $H_{1,\mu-1}$  depend only on $A_\nu$ for $\nu<\mu$. 

First, note that in the block decomposition with respect to $P_*$ (the full spectral projection), 
the $P_*^\perp (\cdots) P_*^\perp$-blocks  on both sides  of \eqref{Solve1} resp.\ \eqref{Solve2}   vanish identically, independently of the choice of $A_\mu$. 
Second, the off-diagonal blocks of \eqref{Solve1} resp.\ \eqref{Solve2} determine~$A_\mu$ uniquely.
Third,   we will  choose  the coefficient $h_\mu$ of the effective Hamiltonian in such a way that also the $P_* (\cdots) P_* $-block of   \eqref{Solve1} resp.\ \eqref{Solve2}  is zero.

We first consider $\mu=1$, which is  special, as we will see, for several reasons.
According to Lemma~\ref{lemma:I1}, the unique solution of the off-diagonal part  of  \eqref{Solve1} is given by   the off-diagonal part of 
   \[
A_1  = - \delta\, \mathcal{I}_{H_0,g,\tilde g} (  K ) + \mathcal{I}_{H_0,g,\tilde g} (  \tfrac{1}{\epsi}V_v) + \mathcal{I}_{H_0,g,\tilde g} (  H_{1,0}) \,.
\]
By assumption we have that $H_{1,0}=H_1 \in \mathcal{L}_{I,\mathcal{S},\infty,L^H_1}$. 
Lemma~\ref{lemma:I2} implies $K = \mathcal{I}_{H_0,g,\tilde g}(\dot H ) \in \mathcal{L}_{\mathcal{S},\infty,L^H_{\epsi^\gamma} }$ and,  since, according to Lemma~\ref{lemma:Vcomm}, $H_{0,v} := \left[ H_0, \tfrac{1}{\epsi}V_v\right] \in \mathcal{L}_{I,\mathcal{S},\infty,L^H_{\epsi^\gamma }}$, also  
$H_{0,v}\in \mathcal{L}_{I,\mathcal{S},\infty,L^H_{\epsi^\gamma }}$. Thus, again by Lemma~\ref{lemma:I2},  $A_1$ 
is  a polynomial  in $\delta$ of degree one with coefficients in 
 $\mathcal{L}_{I,\mathcal{S},\infty,L^H_{\epsi^\gamma}}$ and the constant term 
  $ A_1^{ 0} =  \mathcal{I}_{H_0,g,\tilde g} (  \tfrac{1}{\epsi}V_v+  H_{1,0})$ is the coefficient of the almost-stationary state.\smallskip

We still need to deal with the $P_* (\cdots) P_* $-block  of \eqref{Solve1}. 
Defining
\begin{equation}\label{P0hP0}
P_* h_0 P_*  :=  P_*  (H_0 +V_v)  P_* \qquad\mbox{ and }\qquad P_* h_1P_*  :=   P_* ( H_{1,0}+ \ad_{H_0} (A_1))   P_*\,,
\end{equation}
we have   that the $P_* (\cdots) P_* $-block of   \eqref{Solve1} is
\[
P_*  \left[   \tfrac{1}{\epsi}(h_0 - H_0 -V_v) + (h_1 -  H_{1,0}+ \ad_{H_0} (A_1)), P_n\right] P_* = 0\,.
\]
 To ensure that $h_0$ and $h_1$ are indeed quasi-local operators, we need to add also a $P_*^\perp (\cdots) P_*^\perp $-block
to both of them. Using Corollary~\ref{diagcor}, we find that 
\[
h_0 := H_0 +  V_v -  \mathcal{I}_{H_0 ,g,\tilde g}( \ad_{H_0(t)}(V_v)) 
\]
and
\[
h_1 :=     H_{1,0}- \ad_{H_0} (A_1)     - \mathcal{I}_{H_0 ,g,\tilde g}( \ad_{H_0 }( H_{1,0}- \ad_{H_0} (A_1)   ))
\]
are indeed quasi-local and diagonal,
   and that this new definition is compatible with~\eqref{P0hP0}.  
   
For $\mu>1$ the off-diagonal part of \eqref{Solve2} is again solved by 
\[
A_\mu = \mathcal{I}_{H_0,g,\tilde g}( - \delta Q_\mu + L_\mu + \tilde V_{\mu-1} + H_{1,\mu-1})\,.
\]
Assume as induction hypothesis that for all $\nu<\mu$ and $r\in \N_0$ we already showed that
 $ \frac{\D^r}{\D t^r}A^\delta_\nu$  are polynomials in $\delta$ of degree $\nu$ with coefficients in 
  $  \mathcal{L}_{I,\mathcal{S},\infty,L^H_{\epsi^\gamma}}$. 
Then we can conclude from Lemma~\ref{lemma:ads}  and the respective definitions of $L_\mu$, $Q_\mu$, $\tilde V_{\mu-1}$, and $H_{1,\mu-1}$  that also 
 $ \frac{\D^r}{\D t^r}L_\mu$, $ \frac{\D^r}{\D t^r}Q_\mu$, $ \frac{\D^r}{\D t^r}\tilde V_{\mu-1}$, and $ \frac{\D^r}{\D t^r}H_{1,\mu-1}$ for $r\in\N_0 $ are all polynomials  of order at most $\mu$ in $\delta$ with coefficients 
 in $\mathcal{L}_{I,\mathcal{S},\infty,L^H_{\epsi^\gamma}}$.
  Thus with Lemma~\ref{lemma:I2}  also  $ \frac{\D^r}{\D t^r}A_\mu$ is a  polynomial  of order   $\mu$ in $\delta$ with coefficients in  
 $  \mathcal{L}_{I,\mathcal{S},\infty,L^H_{\epsi^\gamma}}$.
 In summary  we conclude that for $\delta\leq \frac{1}{\epsi}$ we have that $\epsi S = \sum_{\mu=1}^n \epsi^{\mu} A_\mu \in\mathcal{L}_{I,\mathcal{S},\infty,L^H_{\epsi^\gamma}}$.

Finally, the  $P_* (\cdots) P_* $-block of \eqref{Solve2} is canceled by choosing
\begin{align*}
h_\mu = &   -\,   \delta Q_\mu  + L_\mu + \tilde V_{\mu-1} + H_{1,\mu-1}   - \ad_{H_0}( A_\mu)\\
&-  \mathcal{I}_{H_0,g,\tilde g}( \ad_{H_0 }( 
- \delta Q_\mu  + L_\mu + \tilde V_{\mu-1} + H_{1,\mu-1}   -  \ad_{H_0}( A_\mu)
) )\,,
\end{align*}
which is again a diagonal quasi-local operator by Corollary~\ref{diagcor}.

For the remainder term $ R_n = (1+\delta^{n+1})^{-1} U ( Q_{n+1}  - H_{0,n+1} - \tilde V_n -H_{1,n}  ) U^*$ 
 first note that the terms $Q_{n+1}$,   $H_{0,n+1}$, $\tilde V_n$, and $H_{1,n}$ are all of the form of a sum of multi-commutators of operators that are polynomials in $\epsi$ and $(\epsi\delta)$ of total degree at least $n+1$,  conjugated by unitaries of the form $\E^{\I  \tilde \gamma S}$ with $\tilde\gamma\leq \epsi$.  The multi-commutators   can be estimated by Lemma~\ref{lemma:ads}, and the conjugations by Lemma~\ref{lemma:transform}. The conjugation with $U $ that leads to $R_n $ is again estimated by  Lemma~\ref{lemma:transform}.

Finally note that if   $\frac{\D^k}{\D t^k}H (t_0)=0$ for all $k=1,\ldots,n$, then $K(t_0)=0$ and $Q_\mu(t_0) = 0$ for all $\mu=1,\ldots,n$, and thus
 $A^{ \delta}_\mu(t_0)=A^{ 0}_\mu(t_0)$ for all $\mu=1,\ldots,n$.      
\end{proof}

\subsection{Proof of the adiabatic theorem} \label{sec:proofAdiabatic}

  \begin{proof}[Proof of Theorem~\ref{theorem:Adi}.]
To show \eqref{AdiState},
we first observe that  
\[
\tfrac{\epsi^{|\ell_H|\gamma}}{M^{d -|\ell|-|\ell_H|}}    \tr  (\rho^{\epsi,\Lambda,\delta}(t)  B^{\Lambda} )= \tfrac{\epsi^{|\ell_H|\gamma}}{M^{d -|\ell|-|\ell_H|}}   \sum_{X \subset \Lambda}     \tr  \left(\rho^{\epsi,\Lambda,\delta}(t)  \Phi^\Lambda_B(X) \right)  \,. 
\]
We freely use the notation from Proposition~\ref{proposition:Adi} and its proof provided in the last section.
Let  $U^{\epsi,\Lambda,\delta}_n(t,s)$ and  $V^{\epsi,\Lambda,\delta}_n(t,s)$ be the solutions of 
\[
\I\epsi \delta \tfrac{\D}{\D t} U^{\epsi,\Lambda,\delta}_n(t,s) =  H^{\epsi,\Lambda}(t)    \,U^{\epsi,\Lambda,\delta}_n(t,s)  \]
 with $U^{\epsi,\Lambda,\delta}_n(s,s) = {\bf 1}$ for $ t,s\in \R$ and
\[
\I\epsi \delta \tfrac{\D}{\D t} V^{\epsi,\Lambda,\delta}_n(t,s) = \left( H^{\epsi,\Lambda}(t)  + \epsi^{n+1}(1+\delta^{n+1}) R^{\epsi,\Lambda,\delta}_n (t)\right) \,V^{\epsi,\Lambda,\delta}_n(t,s)
  \]
 with $V^{\epsi,\Lambda,\delta}_n(s,s) = {\bf 1}$ for $ t,s\in \R$, respectively, 
and thus 
\[
\rho^{\epsi,\Lambda,\delta}(t) =U^{\epsi,\Lambda,\delta}_n(t,t_0)\,   \Pi^{\epsi,\Lambda,\delta}_n(t_0)\,U^{\epsi,\Lambda,\delta}_n(t_0,t)
\]
and
\[
  \Pi^{\epsi,\Lambda,\delta}_n(t) = V^{\epsi,\Lambda,\delta}_n(t,t_0)\,    \Pi^{\epsi,\Lambda,\delta}_n(t_0)\, V^{\epsi,\Lambda,\delta}_n(t_0,t)\,.
\]
Then for any local observable $O\in \mathcal{A}_X^+$ we have that
\begin{align*} 
&\tr \left( \left(    \Pi^{\epsi,\Lambda,\delta}_n(t) -\rho^{\epsi,\Lambda,\delta}(t)  \right) O\right)\\
&= \tr \left( \left( V^{\epsi,\Lambda,\delta}_n(t,t_0) \Pi^{\epsi,\Lambda,\delta}_n(t_0)V^{\epsi,\Lambda,\delta}_n(t_0,t)- U^{\epsi,\Lambda,\delta}(t,t_0) \Pi^{\epsi,\Lambda,\delta}_n(t_0)U^{\epsi,\Lambda,\delta}(t_0,t) \right)O\right)\\
&=    \int_{t_0}^t\D s\, \tr \Big(  \tfrac{\D}{\D s} \Big( U^{\epsi,\Lambda,\delta}(t_0,s)V^{\epsi,\Lambda,\delta}_n(s,t_0) \Pi^{\epsi,\Lambda,\delta}_n(t_0)V^{\epsi,\Lambda,\delta}_n(t_0,s)U^{\epsi,\Lambda,\delta}(s,t_0)\Big)   \\ &  \hspace{7.5cm}\times U^{\epsi,\Lambda,\delta}(t_0,t)\,O\, U^{\epsi,\Lambda,\delta}(t,t_0)
\Big)\\
&= 
\frac{\I}{\epsi\delta}  \int_{t_0}^t\D s\, \tr \Big(   \left[ \epsi^{n+1}(1+\delta^{n+1}) R_n^{\epsi,\Lambda,\delta}(s),V^{\epsi,\Lambda,\delta}_n(s,t_0) \Pi^{\epsi,\Lambda,\delta}_n(t_0)V^{\epsi,\Lambda,\delta}_n(t_0,s)\right]    \\ &   \hspace{7.5cm}\times  U^{\epsi,\Lambda,\delta}(s,t)\,O\, U^{\epsi,\Lambda,\delta}(t,s)
\Big)\\
&= 
- \I \,\frac{\epsi^{n+1}(1+\delta^{n+1})}{\epsi\delta}  \int_{t_0}^t\D s\, \tr \left(   V^{\epsi,\Lambda,\delta}_n(s,t_0) \Pi^{\epsi,\Lambda,\delta}_n(t_0)V^{\epsi,\Lambda,\delta}_n(t_0,s) \right.\\
&\left.\hspace{5.9cm}\times \left[ R_n^{\epsi,\Lambda,\delta}(s), U^{\epsi,\Lambda,\delta}(s,t)\,O\, U^{\epsi,\Lambda,\delta}(t,s)\right] 
\right).
\end{align*}
Since $ \tr(\Pi^{\epsi,\Lambda,\delta}_n(t_0)) = \tr(\rho^{\epsi,\Lambda,\delta}(t_0)) = \tr(\rho_0^\Lambda)\leq \kappa$,  with $  \sup_{t\in\R} \|\Phi_{R_n}\|_{\zeta_0,0,L^H_{\epsi^\gamma} } \leq C_{R_n}  $ it holds that 
\begin{align}\label{localEst}\nonumber 
\Big| & \tr\left( ( \rho^{\epsi,\Lambda,\delta}(t) -  \; \Pi^{\epsi,\Lambda,\delta}_n(t)  )\, O\right)\Big| \; \leq\\
& \leq \;|t-t_0|  \,\frac{\epsi^{n+1}(1+\delta^{n+1})}{\epsi\delta} \kappa \sup_{s\in[t_0,t]} \left\|\left[ R_n^{\epsi,\Lambda,\delta}(s), U^{\epsi,\Lambda,\delta}(s,t)\,O\, U^{\epsi,\Lambda,\delta}(t,s)\right]\right|\nonumber \\\nonumber
&\leq  C \, C_{R_n}\,   \frac{\epsi^{n+1}(1+\delta^{n+1})}{\epsi\delta} \,  \kappa \, \|O\|\, |X|^2  \,\zeta(\epsi^\gamma\,{\rm dist}(X,L^H)) \,|t-t_0| (1+(\delta\epsi)^{- d}|t-t_0| ^{d})\\
&\leq  C \, C_{R_n}\,   \frac{\epsi^{n+1}(1+\delta^{n+1})}{(\epsi\delta)^{d+1}} \,  \kappa\,  \|O\|\, |X|^2  \,\zeta(\epsi^\gamma\,{\rm dist}(X,L^H)) \,|t-t_0|  (1+ |t-t_0| ^{d})
\,,
\end{align}
where the second inequality follows from Lemma~\ref{lemma:commutator2} (due to the adiaba\-tic time scale we pick up a factor $(\epsi\delta)^{- d}$) and   set ${\rm dist}(X,L^H) := \min_{x \in X} {\rm dist}(x,L^H)$ (compare~\eqref{def:distxL}).
Substituting $\Phi_B^\Lambda (X)$ for~$O$ in \eqref{localEst} and abbreviating $\Delta   := |t-t_0|$, we obtain 
\begin{align}\label{comparerhoPi}
 \Big|  &  \tr    \left(( \rho^{\epsi,\Lambda}(t) -   \Pi^{\epsi,\Lambda}(t)  )  B \right) \Big| \;\leq\nonumber\\& \leq\;  
  \frac{\epsi^{n+1}(1+\delta^{n+1})}{(\epsi\delta)^{d+1}}  \kappa  C \sum_{X \subset \Lambda}   |X|^2\,\zeta(\epsi^\gamma\,{\rm dist}(X,L^H)) \, \|\Phi_B^\Lambda(X)\|\,\Delta   (1+ \Delta  ^{d})  
\nonumber \\
&\leq \;
  \frac{\epsi^{n+1}(1+\delta^{n+1})}{(\epsi\delta)^{d+1}}  \kappa  C   \sum_{x\in \Lambda} \sum_{X \subset \Lambda:\:x\in X}   |X|^2\,\zeta(\epsi^\gamma\,{\rm dist}(x,L^H)) \, \|\Phi_B^\Lambda(X)\|\,\Delta   (1+ \Delta  ^{d})\nonumber\\
&\leq \;
  \frac{\epsi^{n+1}(1+\delta^{n+1})}{(\epsi\delta)^{d+1}}  \kappa  C  \sum_{x\in\Lambda}   \zeta(\epsi^\gamma\,{\rm dist}(x,L^H)) \,  \sum_{y\in \Lambda} F_\zeta(d_{L_1}^\Lambda(x,y)) \,\times \nonumber\\
  &  \hspace{3.5cm}\sum_{X \subset \Lambda:\:x,y\in X}  \mbox{$\Lambda$-diam}(X)^{2d} \, \frac{\|\Phi_B^\Lambda(X)\|}{F_\zeta(d_{L_1}^\Lambda(x,y))}\,\Delta   (1+ \Delta  ^{d})\nonumber\\
&\leq\;  \frac{\epsi^{n+1}(1+\delta^{n+1})}{(\epsi\delta)^{d+1}}   \kappa  C  \|\Phi_B\|_{\zeta,2d,L_1} \,\|F\|_\Gamma \,\times \nonumber\\
  &  \hspace{4.5cm}\sum_{x\in\Lambda}   \zeta(\epsi^\gamma\,{\rm dist}(x,L^H))  \zeta({\rm dist}(x,L)) \,\Delta   (1+ \Delta  ^{d})
\nonumber\\
&\leq \;
  \frac{\epsi^{n+1- |\ell_H|\gamma}(1+\delta^{n+1})}{(\epsi\delta)^{d+1}} 
  \kappa  C  \|\Phi_B\|_{\zeta,2d,L_1} \,\|F\|_\Gamma M^{d-|\ell|-|\ell_H|}\,\Delta   (1+ \Delta  ^{d})
\,.
\end{align}
In the second inequality we used that summing over all sets $X$ for which $x$ minimises the distance to $L^H$ and then over all $x\in\Lambda$ would also include each term in the sum of the previous line at least once.
In the second-to-last inequality we used Lemma~\ref{lemma:dist}.
Hence
the
  statement \eqref{AdiState} of the theorem follows.     
\end{proof}

\subsection{Proof of the expansion of the NEASS}\label{proposition:expand:proof}

 \begin{proof}[Proof of Proposition~\ref{proposition:expand}.] Expanding $\Pi^{\epsi,\Lambda,\delta}$ in the trace yields
\begin{align}\label{PiExpandSA}\nonumber
  \tr  \left(  \Pi^{\epsi,\Lambda,\delta} \hspace{-3pt} \right.&\left.B^{\Lambda}\right) \;=\;   \tr  \left(  \E^{\I  \epsi S^{\epsi, \Lambda,\delta}}  P_*^{\Lambda} \E^{-\I  \epsi S^{\epsi, \Lambda,\delta}}  B^{\Lambda}\right)=   \tr  \left(   P_*^{\Lambda} \E^{-\I \epsi S^{\epsi, \Lambda,\delta}}  B^{\Lambda} \E^{\I  \epsi S^{\epsi, \Lambda,\delta}}\right)\\
&=\;  \tr  \left(   P_*^{\Lambda}  \left(\sum_{j=0}^k \frac{  \epsi ^j}{j!}   \ad_{S^{\epsi}}^j (B)^{\Lambda}
+ \frac{  1}{(k+1)!}  \E^{-\I \tilde\gamma S^{\epsi, \Lambda,\delta}}   \ad_{\epsi S^{\epsi}}^{k+1} (B)^{\Lambda} \E^{\I  \tilde \gamma S^{\epsi, \Lambda,\delta}} 
\right)
\right)
\end{align}
for some $\tilde\gamma\in [0,\epsi]$. 
Using Lemma~\ref{lemma:commutator1}, the argument that took us from \eqref{localEst} to \eqref{comparerhoPi}   gives
\begin{align*}
\left\|  \E^{-\I \tilde\gamma S^{\epsi, \Lambda,\delta}}   \ad_{\epsi S^{\epsi,\Lambda,\delta}}^{k+1} (B)^{\Lambda} \E^{\I  \tilde \gamma S^{\epsi, \Lambda,\delta}} \right\| & \;= \;\left\|
  \ad_{\epsi S^{\epsi,\Lambda,\delta}}^{k+1} (B)^{\Lambda} \right\|\\
  &\; \leq \; C (\epsi^{k+1} +(\epsi\delta)^{k+1})  \|\Phi_B\|_{\zeta,k+1,L}\,M^{d-|\ell|-|\ell_H|} \,.
\end{align*}
Hence the remainder in \eqref{PiExpandSA} is 
 of order $\epsi^{k+1}(1+\delta^{k+1}) M^{d-|\ell|-|\ell_H|}$ and we can truncate it. Note that in the last inequality above we used  the existence of a function $\tilde \zeta\in\mathcal{S}$ with $\mathcal{B}_{\zeta,k+1,L_H} \subset \mathcal{B}_{\tilde \zeta,k+1,L_H}$ such that $S (t) \in \mathcal{L}_{\tilde \zeta, k+1,L_H}$ (compare Lemma~A.1 (a) in \cite{MT}).  
 Expanding also the sum in the second line of \eqref{PiExpandSA} and collecting terms with the same power of $\epsi$ we find that
 \begin{align*}
   \tr  \left(   P_*^{\Lambda}    \E^{-\I \epsi S^{\epsi, \Lambda}}  B^{\Lambda} \E^{\I  \epsi S^{\epsi, \Lambda}}\right) 
 \,=:\,  & \tr  \left(  P_*^{\Lambda} \sum_{j=0}^k \epsi^j \mathcal{K}_j^\Lambda( B^{\Lambda} )\right)\\
 & \,+\,\Or(\epsi^{k+1} (1+\delta^{k+1})\|\Phi_B\|_{\zeta,k+1,L}\,M^{d-|\ell|-|\ell_H|}),
 \end{align*}
 where $\mathcal{K}_j^\Lambda(t): \mathcal{A}_\Lambda^\mathfrak{N} \to  \mathcal{A}_\Lambda^\mathfrak{N}$ are linear maps 
 given be nested commutators with various $A_\mu$'s. In particular,
 \[
 \mathcal{K}_0^\Lambda (B^\Lambda) = B^\Lambda\,,\quad \mathcal{K}_1^\Lambda(B^\Lambda)  = -\I [A^{\Lambda}_1,B^{\Lambda}]\,,
 \]
 and
 \[
 \mathcal{K}_2^\Lambda(B^\Lambda) =-\I [A^{\Lambda}_2,B^{\Lambda}] -\tfrac{1}{2}[A^\Lambda_1, [A^{\Lambda}_1,B^{\Lambda}]  ]\,.
 \]
 For the explicit expression of the linear term we compute
\[
 \tr(P_*   [A _1,B ])  = \tr([ P_*,A_1]B) =  \tr([ P_*,A_1^{\rm OD} ]B)  
\]
and note that 
if  $\sigma_*(t) =\{E_*(t)\}$ is a single eigenvalue,
then for any off-diagonal operator $O=O^{\rm OD}$ one has according to \eqref{Rcommu}
\[
\ad_{H_0}^{-1} (O) =\I\,[(H_0-E_*)^{-1} P_*^\perp, O ]=: \I\,[ R ,O]\,,
\]
and thus with $K^{\rm OD} = -\I [\dot P_*,P_*]$, that
\begin{eqnarray*}
A_1^{\rm OD}   &=&  \mathcal{I}_{H_0,g,\tilde g} ( - \delta K +  \tfrac{1}{\epsi}V_v+  H_{1})^{\rm OD} =  \ad_{H_0}^{-1} (  (- \delta K+\tfrac{1}{\epsi}V_v+  H_{1})^{\rm OD}) \\&=& [ R,  -   \delta [ \dot P_*,  P_*] +\I(\tfrac{1}{\epsi}V_v+  H_{1})^{\rm OD}]\,.
\end{eqnarray*}
 Hence,
\begin{eqnarray*}
 \tr(P_*   [A _1,B ]) &=&   -    \delta\,\tr([ P_*, [ R, [\dot P_*,P_*]]B) + \I \,\tr([ P_*, [ R,  \tfrac{1}{\epsi}V_v+  H_{1} ] ]B)
 \\ &=&  - \delta\,\tr(  P_* [( \dot P_* R + R\dot P_*),B]) - \I\,  \tr( P_* [    [ \tfrac{1}{\epsi}V_v+  H_{1},R],B])
 \,.
\end{eqnarray*}
We skip the longer but in principle similar computation of the explicit second order term. Note however, that the uniqueness of the expansion of $\Pi^{\epsi,\Lambda}_n$ implies  that for computing this expansion one can chose a more convenient series $A_\mu$ not involving the map $\mathcal{I}_{H_0,g,\tilde g}$. For example, the series $\tilde A_\mu$ obtained by choosing $\tilde A_\mu$ purely off-diagonal in each step of the construction leads to major simplifications in computing higher order terms.  While this leads to a different unitary $\E^{-\I \epsi \tilde S^\epsi}$, which presumably does not preserve quasi-locality, they result    in the same NEASS up to terms of order $\epsi^{n+1}$.      
\end{proof}

\appendix
\section*{Appendices}
In the following appendices we collect the various technical results used in the preceding sections. We will make use of the notation established in the previous sections, but for the sake of readability we   often drop the superscripts $\epsi$ and  $\Lambda$.

\section{Proof of Lemma~\ref{lemma:Vcomm}}\label{app:Vlemma}

Before giving the proof of Lemma~\ref{lemma:Vcomm}, we need to briefly discuss the Hilbert space structure of 
$\mathcal{A}_\Lambda$. 
Defining 
$
\langle A, B\rangle_{\mathcal{A}_\Lambda} := \tr A^* B
$
  turns the $4^{|\Lambda| \dimfib}$-dimensional vector space $\mathcal{A}_\Lambda$ into a Hilbert space. A convenient orthonormal basis is formed by the monomials. Let  $\mathfrak{M}_\Lambda:=\{ f : \Lambda\times\{1,\ldots,\dimfib\}\to \{0,1,2,3\} \}$ then
\[
M_f := n_f \prod_{x\in\Lambda} \prod_{j=1}^\dimfib m_{j,x,f(x)}  \quad\mbox{ with } \quad  \left\{\begin{array}{l} m_{j,x,0} = {\bf 1} \\m_{j,x,1} = {a_{j,x}}\\ m_{j,x,2} = {a^*_{j,x}}\\ m_{j,x,3} = {a^*_{j,x} a_{j,x} - a_{j,x}a^*_{j,x}}\,,\end{array}\right.
 \,,
\]
where the ordering with respect to different sites only contributes a sign and is supposed to be fixed in any convenient way. The normalisation factor $n_f$ is   specified below.   
The support of a monomial is the set supp$(M_f):= \{ x\in \Lambda \,|\, f(x,j) \not=0\mbox{ for all } j=1,\ldots,\dimfib \}$.
The $2^{|\Lambda|\dimfib}$ vectors of the form 
\[
\Psi_{X_1,\ldots, X_\dimfib} := \prod_{j=1}^\dimfib \prod_{x\in X_j\subset \Lambda}   a_{j,x}^* \Omega
\]
form an orthonormal basis of Fock space, where, again, the ordering with respect to different sites only contributes a sign that is left unspecified for the moment.  

It suffices to show that the monomials form   an orthonormal basis of Fock space $\mathfrak{F}_{\Lambda,\dimfib=1}$ for the case $\dimfib=1$, as for $\dimfib>1$ we are just dealing with a tensor product of copies of $\mathfrak{F}_{\Lambda,\dimfib=1}$, i.e.\
$\mathfrak{F}_{\Lambda,\dimfib}=\otimes_{j=1}^\dimfib \mathfrak{F}_{\Lambda,\dimfib=1}$.
For two monomials $M_f$ and $M_g$ we have
\[
\tr(M_f^* M_g) = \sum_{X\subset \Lambda} \langle   \Psi_X, M_f^* M_g \Psi_X\rangle\,.
\]
Assume that $f(x) \not= g(x)$ for some $x\in \Lambda$. Up to a sign, the following combinations of creation and annihilation operators on the site $x$ can appear in  $M_f^* M_g$,
\[
a_x\,, \, a_x^*\,, \, a_x a_x \,, \, a_x^* a_x^*\,,  \,a_x^* (a^*_x a_x - a_xa^*_x)\,,\, a_x(a^*_x a_x - a_xa^*_x)\,,(a^*_x a_x - a_xa^*_x)\,.
\]
If one of the first six cases appears, the inner product $ \langle   \Psi_X, M_f^* M_g \Psi_X\rangle$ is zero for any $\Psi_X$. In the last case, the summands with $\Psi_{X\cup\{x\}}$ and with $\Psi_{X\setminus\{x\}}$ cancel because they carry the opposite sign.

In the case   $f=g$,  the following combinations of creation and annihilation operators on any site $x\in\Lambda$ can appear in  $M_f^* M_f$,
\[
{\bf 1}\,, \, a^*_x a_x \,,\, a_x a^*_x \,,\,  (a^*_x a_x + a_xa^*_x) \,.
\]
Thus if $f(x) =  2$ for some $x\in X$ or $f(x) =  1$ for some $x\notin X$, then $ \langle   \Psi_X, M_f^* M_f \Psi_X\rangle=0$. Otherwise $ \langle   \Psi_X, M_f^* M_f \Psi_X\rangle=1$ and therefore the normalisation factor is
\[
n_f :=  2^{-|\{ x\in \Lambda\,|\, f(x)=0\mbox{ \footnotesize or } f(x) =3\}|/2}\,.
\]

 \begin{proof}[Proof of Lemma~\ref{lemma:Vcomm}.]
Because of the simple form of $V_v$, we can define  the interaction of the commutator $  [A, V_v]$ as
\[
\Phi^\Lambda_{ [A, V_v]}(Z) :=  \sum_{x\in\Lambda}\sum_{j=1}^\dimfib \,[ \Phi^\Lambda_{A }(Z ), v^{\epsi,\Lambda}( x) \,a^*_{j,x}a_{j,x}]\,.
\]
Next we expand $\Phi^\Lambda_{A }(Z )$ in the basis of monomials $(M_f)_{f\in\mathfrak{M}_Z^\mathfrak{N}}$ as
\[
\Phi^\Lambda_{A }(Z ) \,=\,\sum_{f  \in\mathfrak{M}_Z^\mathfrak{N}} c_f M_f\,.
\]
Since $a^*_{j,x}a_{j,x}$ has non-vanishing commutator only with $a_{j,x}$ and $a_{j,x}^*$, it holds that
\[
[M_f,  v^{\epsi,\Lambda}( x) \,a^*_{j,x}a_{j,x}] \;=\; \left\{ \begin{array}{cl}
v^{\epsi,\Lambda}( x) M_f & \mbox{ if $ f(x,j) =1$ }\\
- v^{\epsi,\Lambda}(  x) M_f & \mbox{ if $ f(x,j) =2$ }\\
0 & \mbox{ otherwise.}
\end{array}\right.
\]
Now $M_f\in\mathcal{A}_Z^\mathfrak{N}$ implies that for $Z_{f,1} := \{ (z,j)\in Z\times\{1,\ldots,\dimfib\}\,|\, f(z,j) =1 \}$ and  $Z_{f,2}:=\{  (z,j)\in Z\times\{1,\ldots,\dimfib\}\,|\, f(z,j) =2 \}$ it holds that
 $|Z_{f,1}|=|Z_{f,2}|$. Hence we find that
\begin{eqnarray*} \|\Phi^\Lambda_{ [A, V_v]}(Z)\|&=& 
\left\| \sum_{x\in Z} \sum_{j=1}^\dimfib \left[\sum_{f  \in\mathfrak{M}_Z^\mathfrak{N}} c_f M_f,  v^{\epsi,\Lambda}( x) \,a^*_{j,x}a_{j,x}\right] \right\|\\ &= & \left\| \sum_{f  \in\mathfrak{M}_\Lambda^\mathfrak{N}}\left( \sum_{(x,j)\in Z_{f,1}} v^{\epsi,\Lambda}( x) - \sum_{(x,j)\in Z_{f,2}} v^{\epsi,\Lambda}( x) \right) \,c_f\,  M_f \right\|\\
&\leq& \frac{|Z|\,\dimfib}{2} \left( \max_{x\in Z} v^{\epsi,\Lambda}( x) - \min_{x\in Z} v^{\epsi,\Lambda}( x) \right) \|\sum_{f  \in\mathfrak{M}_\Lambda^\mathfrak{N}} c_f M_f\|\\
&\leq & 
  \frac{\mbox{$\Lambda$-diam}(Z)^d\,\dimfib}{2} \;C_v\,\epsi \cdot\mbox{$\Lambda$-diam}(Z) \,   \|\Phi_A^\Lambda(Z)\|\,,
\end{eqnarray*}
where we used in the first inequality  that the monomials $M_f$ form an orthonormal basis.
Defining $\Phi_{A_{v}}^\Lambda := \frac{1}{\epsi}  \Phi^\Lambda_{ [A, V_v]}$, we obtain
\begin{eqnarray*} 
\sum_{Z\supset\{x,y\}} \mbox{$\Lambda$-diam}(Z)^k \frac{\| \Phi^\Lambda_{A_{v}}(Z)  \|}{F_{\zeta_k}(d_{L_\alpha}^\Lambda(x,y))}&=&
\sum_{Z\supset\{x,y\}}  \mbox{$\Lambda$-diam}(Z)^k \frac{\| \frac{1}{\epsi}\Phi^\Lambda_{ [A, V_v]}(Z)\|}{F_{\zeta_k}(d_{L_\alpha}^\Lambda(x,y))}\\
&\leq&
\frac{\dimfib\,C_v}{2}\,\sum_{Z\supset\{x,y\}}  \mbox{$\Lambda$-diam}(Z)^{k+d+1} \frac{\| \Phi^\Lambda_{A}(Z)\|}{F_{\zeta_k}(d_{L_\alpha}^\Lambda(x,y))}\\&\leq&
\frac{\dimfib}{2}\,C_v \,\| \Phi_{A}\|_{\zeta_k,k+d+1,L_\alpha}\,.  \hspace{3cm}    \qedhere
\end{eqnarray*}

\end{proof}

\section{Technicalities on quasi-local operators}\label{app:tech}

In this appendix we collect several results concerning local and quasi-local operators that were used repeatedly in the construction of the almost-stationary states and in the proof of the adiabatic theorem. They are all proven in Appendix~C of \cite{MT}, most of them based on similar lemmas in \cite{BDF}.

We start with a simple lemma that is at the basis of most arguments concerning localisation near $L$. It will not be explicitly used in the following, but it is important for checking that the claimed generalisations of the following results to the case of $L_\alpha$-localisation hold.
\begin{lemma}\label{lemma:dist} 
It holds that 
\[
\sum_{y\in\Lambda} F_\zeta (d^\Lambda_{L_\alpha}(x,y)) \leq \zeta\left( \alpha \,{\rm dist} (x,L)\right)\,\|F\|_\Gamma
\,,
\]
where
\[
\|F \|_\Gamma:=  \sup_{x\in\Gamma}\sum_{y\in\Gamma} F  (d (x,y))     <\infty\,.
\]
\end{lemma}
The   statement for the $d^\Lambda_{L_1}$-distance is proved in Lemma~C.1 in \cite{MT} and the   statement  for the $d^\Lambda_{L_\alpha}$-distance  follows by the same proof.
 
The next lemma shows that the norm of a quasi-local operator-family localised near $L$ grows at most like the volume of $L$.
\begin{lemma}\label{lemma:bound}
Let   $A\in \mathcal{L}_{\zeta,0,L_\alpha}$, then there is a constant $C_\zeta$ depending only on $\zeta$ such that
\[
   \|A^\Lambda\|\leq \alpha^{-|\ell|}\,M^{d-|\ell|}\,C_\zeta \, \|\Phi_A\|_{\zeta,0,L_\alpha} \|F\|_\Gamma
  \,.
\]
\end{lemma}
This follows by by the same proof as the one of Lemma~C.2 in \cite{MT} using Lemma~\ref{lemma:dist}.
We continue with a norm estimate on iterated commutators with quasi-local operator-families, where we use $\add_A(B):= [A,B]$ to denote the adjoint map.
\begin{lemma}\label{lemma:commutator1}
There is a constant $C_k$ depending only on $k\in\N$ such that
for any $A_1  \in \mathcal{L}_{\zeta,kd,L_\alpha}$, $A_2,\ldots,A_k \in \mathcal{L}_{\zeta,kd}$, $X \subset \Lambda$, and   $O\in\mathcal{A}^+_X$ it holds that
\[
\| \add_{A_k^\Lambda} \circ \cdots \circ \add_{A_1^\Lambda} (O)  \|\leq  C_k\,\|O \| \,|X|^k\,   \zeta\left(\alpha\,{\rm dist}(X,L) \right) \, \|\Phi_{A_1}\|_{\zeta,(k-j)d,L_\alpha}\,
\prod_{j=2}^k \|\Phi_{A_j}\|_{\zeta,(k-j)d} .
\]
\end{lemma}
This follows by by the same proof as the one of Lemma~C.3 in \cite{MT} using Lemma~\ref{lemma:dist}.
The next lemma shows that such an iterated commutator of quasi-local operator-families  is itself a quasi-local
  operator-family and that if one of them is $L_\alpha$-localised, then also the iterated commutator is. For the proof see  Lemma~C.4 in \cite{MT} and Lemma~4.5 (ii) in~\cite{BDF}.
\begin{lemma}\label{lemma:ads}
Let $n\in\N_0$, $k\in \N$, $A_0 \in \mathcal{L}_{\zeta,(n+k)d,L_\alpha}$, and $A_1,\ldots, A_k\in \mathcal{L}_{\zeta,(n+k)d}$. Then $\add_{A_k}\cdots\add_{A_1}(A_0) \in \mathcal{L}_{\zeta,nd,L_\alpha}$
and
\[
\| \Phi_{\add_{A_{k}}\cdots\add_{A_{1}}(A_0)}\|_{\zeta,nd,L_\alpha} \leq C_{k,n}\,   \|\Phi_{A_0}\|_{\zeta,(n+k)d,L_\alpha} 
\, \prod_{j=0}^k \|\Phi_{A_j}\|_{\zeta,(n+k)d} 
\]
with a constant $C_{k,n}$ depending only on $k$ and $n$.
In particular, for $A_0 \in \mathcal{L}_{\zeta,\infty,L_\alpha} $ and  $A_1,\ldots, A_k\in \mathcal{L}_{\zeta,\infty} $ also $\add_{A_k}\cdots\add_{A_1}(A_0) \in \mathcal{L}_{\zeta,\infty,L_\alpha}$. \end{lemma}

We also need to control the norm of commutators with time-evolved local operators. This is the content of the next lemma, which is Lemma~C.5 in \cite{MT}, see also Lemma~4.6 in \cite{BDF}.
\begin{lemma}\label{lemma:commutator2}
Let $H\in \mathcal{L}_{a,0}$ generate the dynamics $u^{ \Lambda}_{t,s}$  with Lieb--Robinson velocity $v:= \frac{1}{a} 2^{2d+2} \|F\|_\Gamma \|\Phi_H \|_{a,0}$. Then there exists a constant $C>0$ such that for any $O\in \mathcal{A}_X^+$ with $X\subset \Lambda$, for any $A\in \mathcal{L}_{\zeta,0,L_\alpha}$  and for any $t,s\in \R$   it holds that 
\[
\| [A,u^{ \Lambda}_{t,s}(O)]\| \leq C \|O\| \|\Phi_A\|_{\zeta,0,L_\alpha} |X|^2 \,\zeta(\alpha\,{\rm dist}(X,L_\alpha  )) \,(1+|t-s|^d)\,.
\]

 If $H\in \mathcal{L}_{\tilde\zeta,0}$ for some $\tilde \zeta\in\mathcal{S}$, one still has that for any $T>0$ there exists a constant $C$ such that
 \[
 \sup_{t,s\in [0,T]}\| [A,u^{ \Lambda}_{t,s}(O)]\| \leq C \|O\| \|\Phi_A\|_{\zeta,0,L_\alpha} |X|^2 \,\zeta(\alpha\,{\rm dist}(X,L_\alpha  ))  \,.
 \]
\end{lemma}

The final lemma in this appendix shows that adjoining a quasi-local $L_\alpha$-localised operator-families with a unitary that is itself the exponential of a quasi-local operator-family  yields   a quasi-local and $L_\alpha$-localised operator-family, cf.\ Lemma~C.7 of \cite{MT}.
\begin{lemma}\label{lemma:transform}
Let $S\in \mathcal{L}_{\zeta,0}$ be self-adjoint and let $D \in \mathcal{L}_{\mathcal{S},\infty,L_\alpha}$, i.e.\ there is a sequence $(\tilde\zeta_n)_{n\in\N_0}$ in $\mathcal{S}$ such that $ \|\Phi_D\|_{\tilde \zeta_n,(n+1)d,L_\alpha}<\infty$. Then the family of operators
\[
\left\{A^\Lambda := \E^{-\I S^\Lambda} \,D^\Lambda\, \E^{\I S^\Lambda} \right\}_\Lambda
\]
defines an operator-family $A\in \mathcal{L}_{\mathcal{S},\infty,L_\alpha}$. More precisely, there is a constant $C_{\| \Phi_S\|_{\zeta,0 } }$ depending  on $\| \Phi_S\|_{\zeta,0 }$,   $\zeta$, $(\tilde\zeta_n)_{n\in\N_0}$, and $d$, and a sequence  $(\xi_n)_{n\in\N_0}$ in $ \mathcal{S}$,   such that
\[
\| \Phi_A \|_{\xi_n,n,L_\alpha} \leq C_{\| \Phi_S\|_{\zeta,0 } } \, \|\Phi_D\|_{\tilde \zeta_n,(n+1)d,L_\alpha}
\]
for all $n\in\N_0$. 
\end{lemma}

\section{Quasi-local inverse of the Liouvillian}\label{app:qli}

Let  $\Hi$ be a finite dimensional Hilbert space, $H \in \mathcal{L}(\Hi) =: \mathcal{A}$ be   self-adjoint, $\sigma_*\subset\sigma(H)$ a subset of eigenvalues of $H$ and $P_*$ the corresponding spectral projection.
     The inner product $\langle A,B\rangle := \tr A^* B$ turns the algebra  $\mathcal{A}$ into a Hilbert space that splits into the orthogonal sum $\mathcal{A} = \mathcal{A}^{\rm D}_{P_*} \oplus \mathcal{A}^{\rm OD}_{P_*}$ with respect to $P_*$, 
\[
A = (P_*AP_* + P_*^\perp AP_*^\perp) + (P_*^\perp AP_* + P_* AP_*^\perp)=: A^{\rm D}_{P_*} + A^{\rm OD}_{P_*}\,.
\] 
Then the map  
  \[
  \ad_H: \mathcal{A}\to\mathcal{A}\,,\quad B\mapsto \ad_H(B) :=-\I [H,B]
  \]
  is called the Liouvillian. It maps self-adjoint operators to self-adjoint operators and its restriction $\ad_H|_{\mathcal{A}^{\rm OD}_{P_*}}: \mathcal{A}^{\rm OD}_{P_*}\to \mathcal{A}^{\rm OD}_{P_*}$ to $P_*$-off-diagonal operators is an isomorphism. If $\sigma_*=\{E_*\}$ consists of a single eigenvalue, then its inverse is explicitly given by
  \begin{equation}\label{Rcommu}
  (\ad_H|_{\mathcal{A}^{\rm OD}_{P_*}})^{-1}: \mathcal{A}^{\rm OD}_{P_*}\to \mathcal{A}^{\rm OD}_{P_*}\,,\quad B\mapsto \I\,  \left[ (H-E_*)^{-1}P^\perp_*, B\right]\,,
  \end{equation}
  where $(H-E_*)^{-1}P^\perp_*$ is a bounded operator called the reduced resolvent. 
   See for example the discussion in Appendix~D of \cite{MT} for the simple proofs of all the claims  about $\ad_H|_{\mathcal{A}^{\rm OD}_{P_*}}$.

 In the context of the so-called quasi-adiabatic evolution (called spectral flow in \cite{BMNS}) in \cite{HW,BMNS} an extension $\mathcal{I}_H$ of $(\ad_H|_{\mathcal{A}^{\rm OD}_{P_*}})^{-1}$ to the full algebra $\mathcal{A}$ was constructed in   a way  that  preserves quasi-locality. 
To understand this quasi-local extension of the inverse of the Liouvillian,
first note that for any $\lambda  >0$ one can find a real-valued, odd function
  $\mathcal{W}_{\lambda }\in L^1(\R)$ satisfying 
  \[
  \sup\{|s|^n |\mathcal{W}_{\lambda }(s)|\,|\, |s|>1\}<\infty\qquad\mbox{for all $n\in\N$ },
  \]
   and with   Fourier transform $\widehat{\mathcal{W}}_{\lambda } \in C^\infty(\R)$ satisfying
  \[
\widehat{\mathcal{W}}_{\lambda }(\omega) = \frac{-\I}{\sqrt{2\pi}\omega} \quad\mbox{ for } |\omega|\geq \lambda\qquad\mbox{ and }\qquad   \widehat{\mathcal{W}}_{\lambda }(0)=0\,.
\]
 An explicit function $\mathcal{W}_{\lambda }$ having all these properties  is constructed in \cite{BMNS}. 
 We need a slightly modified version $\mathcal{W}_{\lambda,\tilde \lambda}$ of this function: Let $\lambda>\tilde \lambda >0$ and $\chi_{\lambda,\tilde \lambda}\in C^\infty(\R)$ a real valued  function with $\chi_{\lambda,\tilde \lambda}(\omega) = 0$ for $\omega \in [-\tilde \lambda,\tilde \lambda]$ and $\chi_{\lambda,\tilde \lambda}(\omega) = 1$ for $|\omega|\geq \lambda$. Then $\mathcal{W}_{\lambda,\tilde \lambda}$ defined through its Fourier transform $\widehat{\mathcal{W}}_{\lambda,\tilde \lambda}:=  \chi_{\lambda,\tilde \lambda}\,\widehat{\mathcal{W}}_{\lambda }$ 
 satisfies, in addition to the properties mentioned above for $\mathcal{W}_{\lambda }$, also  $ \widehat{\mathcal{W}}_{\lambda,\tilde \lambda}(\omega)=0$ for all $\omega\in [-\tilde \lambda,\tilde \lambda]$.

\begin{lemma}\label{lemma:I1}
Assume $H$, $\sigma_*$, and $P_*$ as above and let     $g:= {\rm dist}( \sigma_*, \sigma(H)\setminus\sigma_*) > 0$. Then for any $g>\tilde g>0$ the map 
\[
\mathcal{I}_{H,g,\tilde g} : \mathcal{A}\to \mathcal{A}\,, 
\quad \mathcal{I}_{H,g,\tilde g}(A) := \int_\R  \mathcal{W}_{g,\tilde g}(s)  
\,\E^{\I Hs}\,A\,\E^{-\I Hs} \,\D s
\]
satisfies 
\[
\mathcal{I}_{H,g,\tilde g}|_{\mathcal{A}^{\rm OD}_{P_*}} =   \ad_H|_{\mathcal{A}^{\rm OD}_{P_*}}^{-1}
\]
and, if diam$(\sigma_*)\leq \tilde g$,
\[
P_*\,\mathcal{I}_{H,g,\tilde g}(A)\,P_* = 0 \quad\mbox{ for all } A\in \mathcal{A}\,.
\]
\end{lemma}
\begin{proof}
All claims follow immediately by inserting the spectral 
decomposition of $H=\sum_n E_n P_n$  into the definition of~$\mathcal{I}_H$,
\[
\mathcal{I}_{H,g,\tilde g}(A) = \sqrt{2\pi}\,\sum_{n,m} \widehat{ \mathcal{W}}_{g,\tilde g}(E_m-E_n)  P_n\,A\,P_m\,,
\]
and using that for $E_n\in \sigma_*$ and $E_m \in \sigma(H)\setminus \sigma_*$ 
it holds that $|E_m-E_n|\geq g$, i.e.\ $ \widehat{ \mathcal{W}}_{g,\tilde g}(E_m-E_n) = \frac{-\I}{\sqrt{2\pi}(E_m-E_n)}$, 
and that for $E_n,E_m \in \sigma_*$ it holds that 
$|E_m-E_n|\leq  \tilde g$, i.e.\ $ \widehat{ \mathcal{W}}_{g,\tilde g}(E_m-E_n) =0$.     
\end{proof}
 
The  crucial advantage of the map $\mathcal{I}_H$    is that if $H$ admits Lieb-Robinson bounds, then $\mathcal{I}_H$   maps quasi-local operator-families to quasi-local operator-families, 
an observation originating from \cite{HW} and worked out in detail in \cite{BMNS}. To cover also slowly varying potentials $V_v$ that are not contained in  any of the spaces $\mathcal{L}_{\zeta ,k,L^H_{\epsi^\gamma}}$, we need to slightly extend these results.

\begin{lemma}\label{lemma:I2}
Let $H\in \mathcal{L}_{a,0}$ and let $D=\{D^{\epsi,\Lambda}\}$ be an operator-family such that 
$[H,D] \in \mathcal{L}_{\mathcal{S},\infty,L^H_{\epsi^\gamma}}$ for some $L^H\in\loc$ and $\gamma\in\{0,1\}$.  By Lemma~\ref{lemma:ads} this is the case, in particular, if 
$D\in \mathcal{L}_{\mathcal{S},\infty,L^H_{\epsi^\gamma}}$.
Then  
\[
\mathcal{I}_{H,g,\tilde g}(D)^{\epsi,\Lambda}:= \big\{  \mathcal{I}_{H^\Lambda,g,\tilde g}(D^{\epsi,\Lambda})  \big\}_{\Lambda} 
\]
defines an operator-family $\mathcal{I}_{H,g,\tilde g}(D)\in \mathcal{L}_{\mathcal{S},\infty,L_H^\epsi}$. 
\end{lemma}
\begin{proof}
This statement is a slight generalisation of Theorem~4.8 in \cite{BMNS} or Lemma~4.8 in \cite{BDF}, where $D\in \mathcal{L}_{\mathcal{S},\infty,L^H_{\epsi^\gamma}}$ is required. Since the proof is quite subtle and lengthy, we merely explain the small change due to the relaxed assumption on $D$.

By definition of $\mathcal{I}_{H,g,\tilde g}$ and the fact that $\mathcal{W}_{g,\tilde g} \in L^1(\R)$ is odd, we have that
\begin{eqnarray*}
\mathcal{I}_{H,g,\tilde g} (  D) &= & \int_\R \mathcal{W}_{g,\tilde g} (s) \,\E^{\I H s}\, D\,\E^{-\I H_s} \,\D s =  \int_\R \mathcal{W}_{g,\tilde g} (s) \left( \E^{\I H s}\, D\,\E^{-\I H s}  -D\right) \,\D s\\
&=&  \I \int_\R \mathcal{W}_{g,\tilde g} (s) \int_0^s \E^{\I H u}\, \left[ H, D\right]\,\E^{-\I Hu}  \D u\,\D s
\end{eqnarray*}
with $[H,D] \in \mathcal{L}_{\mathcal{S},\infty,L^H_{\epsi^\gamma}}$ by assumption. Now one proceeds exactly as in \cite{BMNS}, where the additional integration in the variable $u$ does not change the arguments at all, since, on the one hand, for the part of the integral where $|s|\leq T$   the bound $\int_0^s \E^{cu}\D u = \tfrac{1}{c}(\E^{cs}-1)$ is as good as $\E^{cs}$ itself and, on the other hand, for the part of the integral where $|s|\geq T$  one uses that also $|s|\mathcal{W}_{g,\tilde g}(s) $ decays faster than any polynomial.     
\end{proof}

 Combining Lemma~\ref{lemma:I1} and Lemma~\ref{lemma:I2}, we obtain the following corollary.
\begin{corollary}\label{diagcor}
Assume    (A1)$_{\timedom,L_H }$  for $H_0$  and let $D \in \mathcal{L}_{\mathcal{S},\infty,L^H_{\epsi^\gamma}}$.
Then
\[
\tilde D := D-    \mathcal{I}_{H_0(t),g,\tilde g}( \ad_{H_0(t)}(D)) 
\]
satisfies $\tilde D \in\mathcal{L}_{\mathcal{S},\infty,L^H_{\epsi^\gamma}}$ and
\[
P_*(t) \tilde D P_*(t) = P_*(t)  D P_*(t) \,,\quad P_*(t) \tilde D P_*(t)^\perp = P_*(t)^\perp \tilde D P_*(t)  =0\,.
\]

\end{corollary}

\end{document}